\documentclass[final,5p,twocolumn]{elsarticle}

\usepackage{amsmath} 
\usepackage{amssymb}  
\usepackage{amsthm}
\usepackage{hyperref}
\usepackage{cleveref}
\usepackage{subcaption}
\usepackage{dsfont}
\usepackage{xcolor}
\newcommand{\sspace}{\mkern9mu}
\newcommand{\nspace}[1]{\mkern#1mu}

\newcommand{\norm}[1]{\left \Vert #1 \right \Vert}

\newtheorem{definition}{Definition}
\newtheorem{theorem}{Theorem}
\newtheorem{corollary}{Corollary}
\newtheorem{lemma}{Lemma}
\newtheorem{proposition}{Proposition}
\newtheorem*{problem}{Problem Statement}
\newtheorem{assumption}{Assumption}
\newtheorem{remark}{Remark}
\theoremstyle{definition}
\newtheorem{example}{Example}

\newtheorem{casestudy}{Case Study}

\newcommand{\Rn}{\mathbb{R}^n}

\newcommand{\R}{\mathbb{R}}

\newcommand{\X}{\mathcal{X}}
\newcommand{\Xhat}{\widehat{\mathcal{X}}}

\newcommand{\U}{\mathcal{U}}
\newcommand{\hb}{h_{\rm b}}

\newcommand{\x}{\mathbf{x}}
\newcommand{\y}{\mathbf{y}}

\newcommand{\xhat}{\hat{\mathbf{x}}}
\newcommand{\xhatDot}{\dot{\hat{\x}}}

\newcommand{\C}{\mathbf{C}}
\newcommand{\Q}{\mathbf{Q}}
\newcommand{\K}{\mathbf{K}}
\newcommand{\A}{\mathbf{A}}
\newcommand{\B}{\mathbf{B}}
\newcommand{\e}{\mathbf{e}}

\newcommand{\Cs}{\mathcal{C}_{\rm S}} 
\newcommand{\Cb}{\mathcal{C}_{\rm B}} 
\newcommand{\Cbi}{\mathcal{C}_{\rm BI}} 
 
\newcommand{\Ci}{\mathcal{C}_{\rm I}}

\newcommand{\Cihat}{\widehat{\mathcal{C}}_{\rm I}}

\newcommand{\phitrue}[2]{\boldsymbol{\phi} (#1, #2)}
\newcommand{\phihat}[2]{\skew{5}\hat{\boldsymbol{\phi}} (#1, #2)}
\newcommand{\phitrueb}[2]{\boldsymbol{\phi}^{o}_{\rm b} (#1, #2)}
\newcommand{\phihatb}[2]{\skew{5}\hat{\boldsymbol{\phi}}^{o}_{\rm b} (#1, #2)}
\newcommand{\phitruefullb}{\phitrueb{\tau}{\x}}
\newcommand{\phihatfullb}{\phihatb{\tau}{\xhat}}
\newcommand{\phitruefullbT}{\phitrueb{T}{\x}}
\newcommand{\phihatfullbT}{\phihatb{T}{\xhat}}

\newcommand{\phitruebE}[2]{\boldsymbol{\phi}_{\rm b} (#1, #2)}
\newcommand{\phihatbE}[2]{\skew{5}\hat{\boldsymbol{\phi}}_{\rm b} (#1, #2)}
\newcommand{\phitruefullbE}{\phitruebE{\tau}{\x}}
\newcommand{\phihatfullbE}{\phihatbE{\tau}{\xhat}}
\newcommand{\phitruefullbET}{\phitruebE{T}{\x}}

\newcommand{\stmhat}{\widehat{\mathbf{\Phi}}_{\rm b}(\tau,\xhat)}
\newcommand{\stmhatT}{\widehat{\mathbf{\Phi}}_{\rm b}(T,\xhat)}

\newcommand{\Ce}{\mathcal{C}_{\rm E}(t)}
\newcommand{\ebar}{\Bar{e}_0}

\newcommand{\ub}{\mathbf{k}_{\rm b}}

\newcommand{\phinb}[2]{\boldsymbol{\phi}_{\rm b} (#1, #2)}

\newcommand{\phinom}{\phinb{\tau}{\mathbf{x}}}
\newcommand{\phinomT}{\phinb{T}{\mathbf{x}}}

\newcommand{\stmnom}{\boldsymbol{\Phi}_{\rm b}(\tau,\mathbf{x})}
\newcommand{\stmnomT}{\boldsymbol{\Phi}_{\rm b}(T,\mathbf{x})}
\newcommand{\jac}{\mathbf{F}_{\rm cl}}

\newcommand{\tb}{\tau}
\newcommand{\Tt}{T}

\newcommand{\xzero}{\mathbf{x}_0}

\newcommand{\w}{\boldsymbol{\omega}}
\newcommand{\what}{\widehat{\boldsymbol{\omega}}}
\newcommand{\J}{\mathbf{J}}

\newcommand{\bLam}{\mathbf{\Lambda}}

\allowdisplaybreaks

\journal{Systems \& Control Letters}

\usepackage{titlesec}

\titlespacing*{\section}
{0pt}  
{2ex}
{0.8ex} 

\titlespacing*{\subsection}
{0pt}
{2ex}
{0.6ex}

\begin{document}

\begin{frontmatter}

\title{Output Feedback Backup Control Barrier Functions: \\ Safety Guarantees Under Input Bounds and State Estimation Error}

\cortext[cor1]{Corresponding author, email:  \texttt{vanwijk@caltech.edu}}
\author[1]{David E. J. van Wijk\corref{cor1}}
\author[2]{Tamas G. Molnar}
\author[3]{Samuel Coogan}
\author[4]{Manoranjan Majji}
\author[1]{Aaron D. Ames}
\author[1]{Joel W. Burdick}

\affiliation[1]{
    organization={Department of Mechanical and Civil Engineering, California Institute of Technology},
    city={Pasadena},
    state={CA},
    postcode={91125},
    country={U.S.A.}
}
\affiliation[2]{
    organization={Department of Mechanical Engineering, Wichita State University},
    city={Wichita},
    state={KS},
    postcode={67260},
    country={U.S.A.}
}
\affiliation[3]{
    organization={Department of Electrical and Computer Engineering, Georgia Institute of Technology},
    city={Atlanta},
    state={GA},
    postcode={30332},
    country={U.S.A.}
}
\affiliation[4]{
    organization={Department of Aerospace Engineering, Texas A\&M University},
    city={College Station},
    state={TX},
    postcode={77840},
    country={U.S.A.}
}

\begin{abstract}
Guaranteeing the safety of controllers is vital for real-world applications, but is markedly difficult when the states are not perfectly known and when the control inputs are bounded. Backup control barrier functions (bCBFs) use predictions of the flow under a prescribed controller to achieve safety in the presence of bounded inputs and perfect state information. However, when only an estimate of the true state is known, this flow may not be precisely computed, as the initial condition is unknown. Furthermore, the true flow evolves using feedback from the estimated state, thus introducing coupling between known and unknown flows. To address these challenges, we propose a technique that leverages an uncertainty envelope centered around the estimated flow and show that ensuring the safety of this envelope guarantees that the true state satisfies the safety constraints. Additionally, we show that in the presence of state uncertainty, using the resulting \textit{Output Feedback Backup Control Barrier Functions (O-bCBFs)}, there always exists a feasible control input that can guarantee the safety of the true state, even in the presence of input constraints.
\end{abstract}

\begin{keyword}
Nonlinear Control \sep
Safety-Critical Control \sep
Output Feedback

\end{keyword}

\end{frontmatter}

\section{Introduction}
Ensuring the safety of autonomous systems is critical to their real-world deployment.
While many approaches exist to make formal guarantees on the safety of such systems, control barrier functions (CBFs) \cite{ames_2017} have become a key tool for constructing safe controllers that enforce the forward invariance of a given safe set.
While CBFs have been successfully employed in many safety-critical scenarios,
most works make the assumption that the true state is available at all times. In practice, this is a difficult assumption to make -- most autonomous systems generate decisions based on an estimate of the true state obtained from a measured output.
Thus, safety-critical control in the presence of state uncertainty (henceforth referred to as safety-critical output feedback control) is a topic that has begun to receive attention in the literature, though comparatively much less than its ``perfect state'' counterpart.

Related works can generally be separated into those which consider stochastic systems, and those which consider deterministic systems. Works considering safe stochastic output feedback controllers often use It$\hat{\rm o}$ calculus and martingale assumptions for developing probabilistic safety conditions of stochastic systems
\cite{clark_2019_sde, yaghoubi_risk-bounded_2021, yaghoubi_risk-bounded_2021-1}. These works form a natural connection to stochastic state observers/estimators and some works specialize to certain forms of estimators (e.g., Kalman filters \cite{vahs2023belief} or particle filters \cite{vahs2024risk}). In contrast to these, \cite{perception_stochastic} considers deterministic dynamical systems with stochastic measurements, and uses learning-based methods to obtain an estimate of the state, and provide probabilistic bounds on that estimate. 

In a deterministic setting, the synthesis of safety-critical controllers falls under the umbrella of ``robustness" to uncertainty. First motivated by the problem of robustness to disturbances, these works provide safety guarantees by accounting for the worst-case instantiation of the disturbance \cite{jankovic_robust_2018,breeden_robust_2023}, see \cite{alan_param_robustCBF_23} for a comparison of methods. Related works reduce the conservatism of such approaches through techniques such as observers \cite{das2025robustcontrolbarrierfunctions,wang2023DOB} or adaptation \cite{lopez_adaptive_robustCBF,robust_adaptive_24}. Motivated by robustness to disturbances, 
works \cite{agrawal_safe_2023_observer,wang2022observer} develop techniques to synthesize safety-critical controllers using output feedback by characterizing estimation error bounds, the rate at which these change, and robustifying against these errors. These approaches leverage knowledge of the estimator/observer dynamics to inform the resulting safe controllers. Other approaches tackle the problem of designing safe output feedback controllers for deterministic systems by robustifying the safety condition against worst-case state uncertainty
\cite{dean2020guaranteeingsafetylearnedperception,lindeman_outputrobustCBF,das2025_RAL_stateerror,nanayakkara2025safety}.

While the above approaches have been shown to be effective in some domains, they all suffer from the assumption that a CBF can be obtained a priori for which a safe control input always exists. For systems with control input bounds, this is a non-trivial problem even for the case of perfect state information. Motivated by this feasibility challenge, backup control barrier functions (bCBFs) \cite{gurriet_online_2018,gurriet_scalable} provide guarantees for the safety of control systems with input bounds by examining the evolution of the system under a pre-verified safeguarding controller, adding a layer of prediction to the safety condition. This technique for developing safety-critical controllers with bounded inputs still results in a tractable optimization problem, and has been successfully applied to many safety-critical scenarios \cite{dunlap2022comparing,hobbs2023rta,ko2024backup,janwani2024learning,rabiee2025soft,rivera2024forward}. 

Unlike the other works, the authors in \cite{cosner_measurement-robust_2021} utilize backup CBFs and provide a degree of robustness to uncertainty in the current state.
However, this work does not account for how the state uncertainty affects future state predictions, nor the feasibility of the resulting controller. 
Incorporating these uncertainty predictions in the control design is the key to studying the controlled invariance properties of safe output feedback controllers, which is a significant new component of the work presented herein.

In contrast to existing works, we present an approach which considers the constructive synthesis of safe sets for which there always exists a feasible output feedback control signal. 
To do this, we examine the forward evolution of the estimated state under a prescribed controller, and compare this with the forward evolution of the true state under output feedback using the same controller. By leveraging an uncertainty envelope centered around the estimated flow, we derive forward invariance conditions for a set defined entirely by the state estimate, and prove that membership of the estimate to this set guarantees the safety of the true state. Further, we show that a feasible control input always exists which can guarantee safety of the true state, even in the presence of input constraints. 
Thus, our approach handles safety for nonlinear input-constrained control systems which only have access to an estimate of the state based on noisy measurements.

The rest of the paper is organized as follows. \Cref{sec:prelims} overviews CBFs and bCBFs using perfect state information. In \Cref{sec:problem} we formalize the problem of designing output feedback safe controllers with input bounds. \Cref{sec:method1} and \Cref{sec:method2} discuss two distinct methods to develop such controllers by examining the forward evolution of an open-loop and a closed-loop state estimate, respectively, under a controller which is robust to output feedback. In \Cref{sec:examples}, we cover the results of our numerical simulations, and we conclude with \Cref{sec:conclusion}.

\section{Preliminaries} \label{sec:prelims}

\subsection{Control Barrier Functions} \label{sec:CBF}

Consider a nonlinear control affine system of the form\footnote{
Because multiple time variables will be used simultaneously, we will exclusively use $\dot{\x}$ to denote differentiation of $\x$ with respect to the time variable $t$, which we henceforth refer to as the \textit{global} time.}
\begin{align} \label{eq:affine-dynamics}
    \dot{\mathbf{x}} = \mathbf{f}(\mathbf{x}) + \mathbf{g}(\mathbf{x})\mathbf{u},
\end{align}
with state ${\mathbf{x} \!\in\! \mathcal{X} \!\subset\! \mathbb{R}^n}$ and control input ${\mathbf{u} \!\in\! \mathcal{U} \!\subseteq \!\mathbb{R}^m}$.
The functions ${\mathbf{f}:\mathcal{X} \rightarrow \mathbb{R}^n}$ and ${\mathbf{g}:\mathcal{X} \rightarrow \mathbb{R}^{n \times m}}$ are continuously differentiable.
For an initial condition ${\mathbf{x}_0 \triangleq \mathbf{x}(0)  \in \mathcal{X}}$, if $\mathbf{u}$ is given by a locally Lipschitz controller ${\mathbf{k}:\mathcal{X} \rightarrow \mathcal{U}}$, ${\mathbf{u}=\mathbf{k}(\mathbf{x})}$, the closed-loop system has a unique solution.

The safety of \eqref{eq:affine-dynamics} can be framed as a problem of ensuring the \textit{forward invariance} of a safe set in the state space.
A set ${\mathcal{C} \!\subset\! \mathbb{R}^n}$ is forward invariant along the closed-loop system if ${\mathbf{x}(0) \!\in\! \mathcal{C} \!\!\implies\! \!\mathbf{x}(t) \!\in\! \mathcal{C},}$ $\forall {t \geq 0}$, meaning that the system maintains safety for all time if it is safe initially.
Consider the safe set $\Cs$ as the 0-superlevel set of a continuously differentiable function ${h : \mathcal{X} \rightarrow \mathbb{R}}$ with
\begin{align}\label{eq:safety_CS}
    \Cs \triangleq \{\mathbf{x} \in \mathcal{X} : h(\mathbf{x}) \ge 0\},
\end{align}
where the gradient of $h$ along $\partial\Cs$ remains nonzero.

When the state $\x$ is known, the conditions for forward invariance have been well-studied in the context of CBFs.
A function ${h : \mathcal{X} \rightarrow \mathbb{R}}$ is a CBF \cite{ames_2017} for \eqref{eq:affine-dynamics} on $\Cs$ if there exists a class-$\mathcal{K}_{\infty}$ function\footnote{The function ${\alpha : \mathbb{R}_{\ge 0} \rightarrow \mathbb{R}_{\ge 0}}$ is of class-$\mathcal{K}_{\infty}$ if it is continuous, strictly increasing, $\alpha(0)=0$, and $\text{lim}_{x \rightarrow \infty} \nspace{2}\alpha(x) = \infty$.} $\alpha$ such that for all ${\mathbf{x} \in \Cs}$, $\sup_{\mathbf{u} \in \mathcal{U}} \dot{h}(\mathbf{x},\mathbf{u})
    >  -\alpha(h(\mathbf{x}))$ holds.
\begin{theorem}
[Forward Invariance Condition \cite{ames_2017}]
\label{thm: cbf}
If $h$ is a CBF for \eqref{eq:affine-dynamics} on $\Cs$, then any locally Lipschitz controller ${\mathbf{k}:\mathcal{X} \rightarrow \mathcal{U}}$, ${\mathbf{u}=\mathbf{k}(\mathbf{x})}$ satisfying 
\begin{align} \label{eq: cbf_condition}
    \nabla h(\mathbf{x}) \big( \mathbf{f}(\mathbf{x}) + \mathbf{g}(\mathbf{x}) \mathbf{u} \big) \ge -\alpha(h(\mathbf{x})),
\end{align}
for all ${\mathbf{x} \in \Cs}$ renders the set $\Cs$ forward invariant.
\end{theorem}
This condition can be leveraged to design a \textit{safety filter} which seeks to minimize deviations from a primary controller ${\mathbf{k}_{\rm p} : \mathcal{X} \rightarrow \mathcal{U}}$, whilst ensuring the safety of \eqref{eq:affine-dynamics}. The safety filter takes the form of a quadratic program (QP):
\begin{align} 
    \mathbf{k}_{\rm safe}(\mathbf{x}) = \underset{\mathbf{u} \in \mathcal{U}}{\text{arg\,min}} \mkern9mu &
    \left\Vert \mathbf{k}_{\rm p}(\x)-\mathbf{u}\right\Vert^{2}
    \tag{CBF-QP} \label{eq:cbf-qp} \\
    \text{s.t.} \quad &
    \nabla h(\mathbf{x}) \big( \mathbf{f}(\mathbf{x}) + \mathbf{g}(\mathbf{x}) \mathbf{u} \big) \ge -\alpha(h(\mathbf{x})).
    \nonumber
\end{align}
This QP is feasible if the safety function $h$ is a CBF.
While the \eqref{eq:cbf-qp} is a highly useful tool for set invariance, obtaining a CBF which satisfies~\eqref{eq: cbf_condition} for high-dimensional systems with a bounded input set $\mathcal{U}$ is nontrivial. In light of this downside to traditional safety filtering techniques, we review a technique which allows for the automatic construction of CBFs.

\subsection{Backup Control Barrier Functions} \label{sec:bCBF}

Backup CBFs (bCBFs) provide a constructive method to generating controlled invariant sets online for nonlinear systems with input constraints \cite{gurriet_online_2018, gurriet_scalable}.  
A set $\mathcal{C} \!\subseteq\! \R^n$ is \textit{controlled invariant} if there exists a controller ${\mathbf{k}\!:\!\mathcal{X} \!\to\! \mathcal{U}}$, ${\mathbf{u}\!=\!\mathbf{k}(\x)}$ rendering $\mathcal{C}$ forward invariant for \eqref{eq:affine-dynamics}, with ${\mathbf{u} \in \mathcal{U}}$. 
As implied by its name, the bCBF method relies on a backup set and a backup controller, which we define below.
\begin{definition}[Backup Set and Controller]
\label{def: backup}
    A set $\Cb  \triangleq \{\mathbf{x}\in \mathcal{X}  :   h_{\rm b}(\mathbf{x}) \ge 0\}$ is a {backup set} if ${\Cb \subseteq \Cs}$, $\Cb$ is controlled invariant, and ${\nabla h_{\rm b}(\x) \neq \boldsymbol{0},\forall \nspace{2} \x \in \partial\mathcal{C}_{\rm B}}$. A continuously differentiable control law ${\mathbf{k}_{\rm b}: \mathcal{X} \rightarrow \mathcal{U}}$ rendering $\Cb$ forward invariant is called a {backup controller}.
\end{definition}

The benefit of backup CBFs comes from the observation that finding a small controlled invariant subset of $\Cs$ (i.e., the backup set) is often substantially easier than verifying that $\Cs$ itself is controlled invariant. For instance, as discussed in \cite{gurriet_scalable}, a backup set can often be defined by a level set of a quadratic Lyapunov function centered on a stabilizable equilibrium point, and can be rendered forward invariant by a feedback linearization controller. 
The interested reader is referred to \cite{gacsi2025braking} for constructive techniques to generate valid backup sets and controllers.

Having obtained a backup set and controller, one could directly enforce safety using a \eqref{eq:cbf-qp} for $h_{\rm b}$ since ${\Cb \subseteq \Cs}$. However, because $\Cb$ is typically much smaller than $\Cs$, this approach would lead to very conservative behavior.
Instead, to reduce conservatism, the backup CBF method expands the backup set using finite-time reachability by examining the evolution of the backup system:
\begin{align}\label{eq: f_cl}
    \dot{\x} = \mathbf{f}_{\rm cl}(\x) \triangleq \mathbf{f}(\x) + \mathbf{g}(\x)\ub(\x).
\end{align}
Specifically, the backup dynamics~\eqref{eq: f_cl} are forward integrated over a finite horizon, and if the system can safely reach the backup set from ${\x}$ using $\ub$, this state is classified as safe, even if ${\x \notin \Cb}$.
This allows system \eqref{eq:affine-dynamics} to evolve beyond the backup set and results in a controlled invariant set which is larger than $\Cb$.
The expanded controlled invariant set is denoted ${\Cbi \!\subseteq\! \Cs}$ and defined as
\begin{align} \label{def:C_BI}
    \Cbi \triangleq \left\{ \x \in \mathcal{X} \,\middle|\, 
    \begin{array}{c}
    h(\phinom) \geq 0, \forall \nspace{1} \tau \in [0,T], \\
    h_{\rm b}(\phinomT) \geq 0 \\
    \end{array}
    \right\}.
\end{align}
Here, $\phinom$ is the \textit{flow} of  system~\eqref{eq: f_cl} over the interval $\tau\in[0,T]$ for a fixed, finite horizon ${T > 0}$ starting at $\x$:
\begin{align} \label{eq: nomFlow}
    \tfrac{\partial}{\partial \tau}{\boldsymbol{\phi}_{\rm b}}(\tau,\x) = \mathbf{f}_{\rm cl}(\phinom), \sspace \phinb{0}{\x} = \x.
\end{align}
Thus, the first inequality in~\eqref{def:C_BI} states that, starting from $\Cbi$, the backup system~\eqref{eq: f_cl} evolves safely over the horizon ${\tau \in [0,T]}$, while the second inequality in~\eqref{def:C_BI} expresses that the backup system reaches the backup set within time $T$.
Importantly, the expanded set $\Cbi$ is controlled invariant.
\begin{lemma}[Controlled Invariance of $\Cbi$ \cite{gurriet_scalable}{\cite[Lem.~1]{tamasROM2023}}] \label{lemma: CBI_controlledInv}
The set $\Cbi$ is controlled invariant, and the backup controller $\ub$ renders $\Cbi$ forward invariant along \eqref{eq: f_cl} as
${\x \in \Cbi \implies \phinom \in \Cbi \subseteq \Cs, \forall \tau \geq 0}$.
\end{lemma}

Because a controller exists which renders $\Cbi$ forward invariant (namely, the backup controller $\ub$), we can guarantee the safety of \eqref{eq:affine-dynamics} in the presence of input bounds by rendering the implicit set $\Cbi$ forward invariant.
By \Cref{thm: cbf}, a controller ${\mathbf{k}:\mathcal{X} \rightarrow \mathcal{U}}$, ${\mathbf{u}=\mathbf{k}(\mathbf{x})}$ makes $\Cbi$ forward invariant, and thus~\eqref{eq:affine-dynamics} is safe w.r.t.~$\Cs$, if there exist class-$\mathcal{K}_{\infty}$ functions $\alpha$, $\alpha_{\rm b}$ such that
\begin{subequations}\label{eq: nom_bcs1}
\begin{align} 
    \nabla h(\phinom)\stmnom\dot{\mathbf{x}} &\ge - \alpha (h(\phinom)), \label{eq: htraj_nom}  \\ 
    \!\!\nabla h_{\rm b}(\phinomT) \stmnomT \dot{\mathbf{x}} &\ge - \alpha_{\rm b} (h_{\rm b}(\phinomT)), \label{eq: hb_nom}
\end{align}
\end{subequations}
for all ${\tau \in [0,T]}$ and ${\x\in\mathcal{\Cbi}}$. Here, ${\dot{\mathbf{x}}= \mathbf{f}(\mathbf{x}) + \mathbf{g}(\mathbf{x})\mathbf{u}}$, and ${\stmnom\! \triangleq\! \partial \phinom/\partial \mathbf{x}}$ is the state-transition matrix (STM) capturing the sensitivity of the flow to perturbations in the initial state $\mathbf{x}$. The STM is the solution to
\begin{equation} \label{eq: stm_nominal}
  \begin{gathered} 
    \tfrac{\partial}{\partial \tau}{{\boldsymbol{\Phi}}}_{\rm b}(\tau, \mathbf{x}) \!=\! \jac(\phinom)\stmnom,
    \sspace
    \boldsymbol{\Phi}_{\rm b}(0,\mathbf{x}) = \mathbf{I}_n,
\end{gathered}
\end{equation}
where ${\jac(\mathbf{x})\triangleq {\partial \mathbf{f}_{\rm cl}(\mathbf{x})} / {\partial \mathbf{x}}}$
is the Jacobian of $\mathbf{f}_{\rm cl}$ in~\eqref{eq: f_cl}, and $\mathbf{I}_n$ is the ${n \!\times\! n}$ identity matrix.

The constraint on the flow given by \eqref{eq: htraj_nom} must hold for all $\tau$ on the set $[0,T]$, and so in practice this constraint is typically enforced for a discrete selection\footnote{The condition \eqref{eq: htraj_nom} can be can robustified between sample times by applying \cite[Thm.~1]{gurriet_scalable}.} of times in $[0,T]$.
Then,
safety can again be enforced with a QP:
\begin{align*} 
    \mathbf{k}_{\rm safe}(\mathbf{x}) = \underset{\mathbf{u} \in \mathcal{U}}{\text{arg\,min}} \mkern9mu &
    \left\Vert \mathbf{k}_{\rm p}(\x)-\mathbf{u}\right\Vert^{2} \quad \tag{bCBF-QP} \label{eq:bcbf-qp} \\
    \text{s.t.  } 
    & \eqref{eq: htraj_nom}, \ \eqref{eq: hb_nom},
\end{align*}
for all ${\tau \!\in \!\{0, \Delta, \dots, \Tt \}}$ where ${\Delta\! >\! 0}$ is a discretization time step satisfying ${T/\Delta \in \mathbb{N}}$.
In contrast to the \eqref{eq:cbf-qp}, if the backup controller is designed such that ${\ub(\mathbf{x}) \!\in\! \mathcal{U}}$ for all ${\mathbf{x} \!\in\! \Cbi}$ then the feasibility of the~\eqref{eq:bcbf-qp} is guaranteed over $\Cbi$ for appropriate $\alpha$ and $\alpha_{\rm b}$ \cite{gurriet_scalable}\cite[Lemma 2]{tamasROM2023}.

Clearly, bCBFs are a powerful and scalable tool for designing safe controllers with input bounds for complex safety specifications and system dynamics. However, the assumption that is implicit in both \Cref{sec:CBF} and \Cref{sec:bCBF} is that the state $\x$ is perfectly known. For any real-world application, the state will never be perfectly known. In reality, feedback controllers 
utilize an \textit{estimate} of the true state $\x$, which is typically generated by processing noisy sensor measurements. When the estimated and true state differ, the effects on the behavior of \eqref{eq:affine-dynamics} may be drastic, and the interplay of such systems must be studied carefully. Thus, for the remainder of the manuscript, we will focus on generating safe controllers 
which obey input bounds, and are informed only by an estimate of the state.

\section{Problem Formulation} \label{sec:problem}

Having laid the foundation for safety of dynamical systems using control barrier functions with perfect state information, we turn to the problem of guaranteeing safety using output feedback for systems with uncertain state information.
Consider the nonlinear control affine system
\begin{subequations} \label{eq:truth_eqtns}
\begin{align} \label{eq: truth}
    \dot{\mathbf{x}} &= 
    \mathbf{f}(\mathbf{x}) + \mathbf{g}(\mathbf{x})\mathbf{u}, \\
    \mathbf{y} &= \mathbf{z}(\mathbf{x}) + \mathbf{v}(t), \label{eq:measFun}
\end{align}
\end{subequations}
where ${\mathbf{x} \in \mathcal{X} \subset \mathbb{R}^n}$ represents the unknown, \textit{true} state vector, ${\mathbf{u} \in \mathcal{U} \subseteq \mathbb{R}^m}$ is the control input, and ${\y \in  \mathcal{Y} \subseteq \mathbb{R}^{y}}$ is a measured output corrupted by measurement noise ${\mathbf{v}(t) \in \mathbb{R}^{y}}$. We assume that $\mathcal{X}$ is compact and that the admissible input set $\mathcal{U}$ is an $m$-dimensional convex polytope. The functions ${\mathbf{f} : \mathcal{X} \rightarrow \Rn }$ and ${\mathbf{g} : \mathcal{X} \rightarrow \mathbb{R}^{n \times m}}$ are continuously differentiable, and the measurement function ${\mathbf{z} : \mathcal{X} \rightarrow \mathbb{R}^{y}}$ is locally Lipschitz continuous.

\begin{assumption}[Bounded Measurement Noise] \label{ass: measNoise}
The measurement noise is deterministic, piecewise continuous in time, and bounded such that ${\norm{\mathbf{v}(t)} \!\leq\! \Bar{v} \!\in\! \R_{>0}}$ for all ${t \!\geq\! 0}$. 
\end{assumption}

In estimation theory, it is typical for the dynamics model of the true state \eqref{eq: truth} to include a process disturbance representing model uncertainty. In our related works \cite{vanwijk2025_dobcbf,vanWijk_DRbCBF24}, we address model uncertainty in the context of the backup CBF method by robustifying against future state uncertainty caused by bounded disturbances. However, these works assume that the true state is perfectly known, which simplifies the problem since the developed safety-critical controllers may use state feedback.
To better distinguish from these works,
we specifically do not include a process disturbance, and instead focus on addressing state uncertainty through safety-critical output feedback control.

\subsection{Error-bounded State Estimators}

Consider also a \textit{state estimator\footnote{Throughout the text, the terms estimator and observer are used interchangeably to describe a process by which an estimate of the true state may be obtained from measurements.}} which continuously estimates the unknown state $\x$ through a ``prediction'' term and a ``correction'' term. The prediction term duplicates the dynamics \eqref{eq: truth} and the correction term compares the expected measurement using the measurement model, $\mathbf{z}$, to the actual measurement, $\mathbf{y}$, to steer the estimate towards the true state. Such estimators take the form
\begin{equation}
    \dot{\hat{\x}} = \mathbf{f}(\hat{\x}) + \mathbf{g}(\hat{\x})\mathbf{u} + \mathbf{r}(t,\xhat,\y), \label{eq:est} \\
\end{equation}
where ${\xhat\in \widehat{\mathcal{X}} \subset \mathbb{R}^n}$ represents the \textit{estimated} state vector and ${\widehat{\X}}$ is assumed to be compact. The continuously differentiable function 
${\mathbf{r}: \mathbb{R}_{\geq 0} \times \Xhat \times \mathcal{Y} \rightarrow \mathbb{R}^{n}}$ 
captures the correction term; see Example~\ref{ex:linear_correction} below.

Because the true state $\x$ is unknown, the estimated state $\hat{\x}$ is used by a feedback controller ${\mathbf{k} : \widehat{\X} \times \mathbb{R}_{\geq 0} \rightarrow \mathcal{U}}$,
\begin{equation} 
    \mathbf{u} = \mathbf{k}(\xhat,t), \label{eq:feedbackController}
\end{equation}
which is assumed to be Lipschitz continuous in ${\xhat}$ and piecewise continuous in ${t}$.
This leads to the closed-loop output feedback system given by
\begin{subequations}\label{eq:cl_outputFeedback}
\begin{align}
    &\dot{\mathbf{x}} = \mathbf{f}(\mathbf{x}) + \mathbf{g}(\mathbf{x})\mathbf{u}, \label{eq:true-CLsystem} \\ 
    &\mathbf{y} = \mathbf{z}(\mathbf{x}) + \mathbf{v}(t), \\
    &\dot{\hat{\x}} = \mathbf{f}(\hat{\x}) + \mathbf{g}(\hat{\x})\mathbf{u} + \mathbf{r}(t,\xhat,\y), \label{eq:est-CLsystem} \\ 
    & \mathbf{u} = \mathbf{k}(\xhat,t), \\
    &\x(t_0) = \x_0, \nspace{4} \xhat(t_0) = \xhat_0. \label{eq:ic_clsystemch6}
\end{align}
\end{subequations}
For initial conditions \eqref{eq:ic_clsystemch6}, systems \eqref{eq:true-CLsystem} and \eqref{eq:est-CLsystem}
have unique solutions ${\phitrue{t}{\xzero}}$ and ${\phihat{t}{\xhat_0}}$ over an interval of existence, respectively. For the rest of the paper, we take ${{t_0 = 0}}$ without loss of generality.

To establish a well-posed safety-critical output feedback control problem, the output feedback system \eqref{eq:truth_eqtns} is assumed to be locally observable such that the observer can uniquely reconstruct system states from measurements.
Furthermore, we restrict ourselves to \textit{error-bounded estimators}, as defined below, for the error ${\e_x(t) \triangleq \x(t) - \xhat(t)}$.

\begin{definition}[Error-Bounded Estimator] \label{def: error_estimator}
The state estimator in \eqref{eq:est} is called an error-bounded estimator for \eqref{eq:truth_eqtns} and \eqref{eq:feedbackController},  
if there exists a function ${\delta_x : \mathbb{R}_{\geq 0} \times \mathbb{R}_{>0} \rightarrow \mathbb{R}_{\geq 0}}$ satisfying ${\delta_x(t, \Bar{e}_0) \geq \norm{\e_x(t)}}$ for all ${t\!\geq\!0}$ with initial uncertainty ${\norm{\x_0 - \xhat_0} \leq \Bar{e}_{0} \in \R_{>0}}$, where ${\x(t)}$ and ${\xhat(t)}$ are solutions to \eqref{eq:true-CLsystem} and \eqref{eq:est-CLsystem} respectively.
\end{definition}

\begin{remark}[Classes of Error-Bounded Estimators] \label{rem: bounded_error}
    Many related works consider a form of error-bounded estimators (see e.g., \cite{agrawal_safe_2023_observer,dean2020guaranteeingsafetylearnedperception,cosner_measurement-robust_2021,confidenceAware_2024}). For a deterministic extended Kalman filter it is possible to bound the estimation error using Lyapunov stability arguments
    \cite[Ch. 11.2]{khalil_nonlinear_2015}\cite{REIF19981119}. In the case where the observer can be made contracting \cite[Lemma 5.22]{contraction_bullo} or input-to-state stable \cite{iss_og_sontag,nonlinearObs_ISS2016}, the error bound can be well described and it decays (exponentially) with a known rate. 
\end{remark}

We provide an example of an estimator for which estimation error bounds are well-studied in the literature.

\begin{example}[Linear Correction] \label{ex:linear_correction}
Many popular state estimators and their related variants (e.g., Kalman filters, Luenberger observers, Gaussian mixture filters) produce an estimate of the state using a linear correction term. 
Such estimators take the form
\begin{align}
\dot{\hat{\x}} &= \mathbf{f}(\hat{\x}) + \mathbf{g}(\hat{\x})\mathbf{u} + \mathbf{L}(t)\big(\mathbf{y} - \mathbf{z}(\hat{\x})\big), \label{eq:est_Luen} 
\end{align}
where ${\mathbf{L}(t) \in \mathbb{R}^{n \times y}}$ is the estimator gain (for example, in the case of the Kalman filter, this is the Kalman gain).
\end{example}

\subsection{Output Feedback Safety}

In this manuscript, we are interested in guaranteeing the safety of the true system using controllers which rely only on estimated states (typically referred to as \textit{output feedback}) and satisfy input constraints. Because safety is only important with respect to the true system, the safe set is a function of the true state, as defined by $\Cs$ in \eqref{eq:safety_CS}
for a continuously differentiable function ${h : \X \rightarrow \R}$. For output feedback controllers, we must obtain a new definition of forward invariance, which we adapt from \cite{agrawal_safe_2023_observer}. 
\begin{definition}[Output Feedback Safety]\label{def:outputSafety}
    A system \eqref{eq:true-CLsystem} is said to be safe with respect to a set ${\Cs}$ if the closed-loop output feedback system \eqref{eq:cl_outputFeedback} satisfies
    \begin{align}
        \x_0 \in \X_0, \xhat_0 \in \widehat{\X}_0 \implies \x(t) \in \Cs, \nspace{2} \forall \nspace{1} t \geq 0,
    \end{align}
    for initial condition sets ${\X_0,\widehat{\X}_0 \subset \Cs}$.
\end{definition}

Equipped with this modified definition of safety, we now formalize the main problem that is addressed in this paper.

\begin{problem}
    We consider the problem of obtaining conditions on the output feedback controller \eqref{eq:feedbackController} for system \eqref{eq:truth_eqtns} with an error-bounded estimator \eqref{eq:est} under Assumption~\ref{ass: measNoise}
    such that safety w.r.t $\Cs$ is satisfied in the sense of \Cref{def:outputSafety}.
\end{problem}

To obtain these conditions, we extend the framework of backup CBFs.
We leverage a measure of recursive feasibility, endowed by propagating the dynamical systems under a prescribed backup controller, as in \cite{gurriet_scalable}.

\subsection{Backup Set and Controller with Output Feedback}

Consider a backup set ${\Cb \triangleq \{\mathbf{x} \in \mathcal{X} : h_{\rm b}(\mathbf{x}) \ge 0\} \subseteq \Cs}$ for a continuously differentiable function ${\hb: \X \rightarrow \R}$.
Assume that the state estimation error can be bounded in the small region of the state space described by ${\Cb}$. This may correspond to regions of the state space that have favorable observability properties, or where the measurement error is small. For example, in the case of a vision-based estimator, this could correspond to a feature-rich location. 
\begin{assumption}[Estimation Error Bounds in $\Cb$] \label{ass: eb_inCB}
    There exists ${\Bar{e}_{\rm b} \in \mathbb{R}_{>0}}$ such that for all ${\x(t) \in \Cb}$, the estimation error ${\mathbf{e}_x(t)}$ satisfies ${\norm{\e_x(t)} \leq \Bar{e}_{\rm b}}$ for all ${t \geq 0}$, where ${\x(t)}$ and ${\xhat(t)}$ are solutions to \eqref{eq:true-CLsystem} and \eqref{eq:est-CLsystem} respectively.
\end{assumption}

As in \Cref{sec:bCBF}, we are interested in obtaining a backup controller which can render $\Cb$ forward invariant. Since in this case the state is unknown, the backup controller should render $\Cb$ invariant for \eqref{eq: truth} using output feedback. However, designing output feedback controllers directly presents a distinctly more difficult problem than designing state feedback controllers. Therefore, we assume that we can design a simple \textit{state} feedback controller, which has desirable properties with respect to ${\Cb}$. We will show that this assumption has strong implications for its output feedback counterpart.

\begin{assumption}[Robustly Forward Invariant Backup Controller] \label{ass: robust_inv_inputdist}
    The continuously differentiable backup controller ${\ub: \mathcal{X} \rightarrow \mathcal{U}}$ under \textit{state} feedback renders the backup set ${\Cb}$ forward invariant along 
    \begin{align*}
        \dot{\mathbf{x}} = \mathbf{f}(\mathbf{x}) + \mathbf{g}(\mathbf{x})\mathbf{k}_{\rm b}(\x) + \mathbf{d}_{v},
    \end{align*}
    for any additive virtual disturbance ${\mathbf{d}_{v}}$ and $\x \in \Cb$ satisfying ${{\norm{\mathbf{d}_v} \leq \mathcal{L}_{\ub} \norm{\mathbf{g}(\x)} \Bar{e}_{\rm b}}}$ where ${\mathcal{L}_{\ub}}$ is the Lipschitz constant of ${\ub}$ and $\Bar{e}_{\rm b} \in \mathbb{R}_{>0}$ satisfies \Cref{ass: eb_inCB}.
\end{assumption}
We note that because \Cref{ass: robust_inv_inputdist} considers state feedback controllers, the backup controller can be designed to be robust to additive disturbances 
via traditional techniques such as by robustifying the level sets of Lyapunov functions \cite{jankovic_robust_2018},
\cite[Ch. 3]{freeman_robust_1996},
\cite[Ch. 13.1]{khalil2002nonlinear} or see \cite{vanWijk_DRbCBF24,vanwijk2025_dobcbf}. This assumption could also be satisfied by designing an input-to-state safe (ISSf) backup controller \cite{Issf_ames}. We note that \Cref{ass: robust_inv_inputdist} may be more difficult to satisfy than in related works due to the dependence on the Lipschitz constant of the backup controller (i.e., changing the controller will change the magnitude of the virtual disturbance). However, in \Cref{sec:examples} we demonstrate how to derive a backup controller satisfying this assumption for a nonlinear spacecraft rotation example. Further, in \Cref{sec:LinearCH6}, we derive backup controllers for which \Cref{ass: robust_inv_inputdist} is no longer required (see \Cref{prop:linearsystemCBandub}).

We now show that the backup controller designed to satisfy \Cref{ass: robust_inv_inputdist} under state feedback renders ${\Cb}$ invariant for the true system under \textit{output} feedback.
\begin{lemma}[Forward Invariance of $\Cb$ under Output Feedback] \label{lemma: cb_fwd_inv}
    A controller ${\ub : \X \rightarrow \U}$ satisfying \Cref{ass: robust_inv_inputdist} renders ${\Cb}$ forward invariant along \eqref{eq:true-CLsystem} with ${\mathbf{u} = \ub(\xhat)}$ such that ${\x_0 \in \Cb \implies \x(t) \in \Cb, \nspace{1} \forall \nspace{1} t \geq 0}$.
    \begin{proof}
        Examining the output feedback system \eqref{eq:true-CLsystem} under ${\ub(\xhat)}$ and adding zero we obtain
        \begin{align*}
            \dot{\x} 
            &= \mathbf{f}(\x) + \mathbf{g}(\x)\ub(\xhat) \\[-1.5em]
            &= \mathbf{f}(\x) + \mathbf{g}(\x)\ub(\x) + \overbrace{\big( \mathbf{g}(\x) \ub(\xhat) - \mathbf{g}(\x) \ub(\x) \big)}^{\mathbf{d}_{v}}.
        \end{align*}
        Using the Cauchy–Schwarz inequality, the virtual disturbance ${\mathbf{d}_{v}}$ can be bounded as
        $\norm{\mathbf{d}_{v}} \!\leq\! \nspace{1}\mathcal{L}_{\mathbf{k}_{\rm b}} \norm{\mathbf{g}(\x)}\norm{\xhat \!-\! \x}$.
        Since ${\x_0 \in \Cb}$, \Cref{ass: eb_inCB} implies ${\norm{\xhat(t) - \x(t)} \leq \Bar{e}_{\rm b}}$ ${\forall \nspace{2} t \geq 0}$, so applying \Cref{ass: robust_inv_inputdist} completes the proof.
    \end{proof}
\end{lemma}

\subsection{Backup Flows with Output Feedback}

Next, we seek to expand the backup set to a larger set, similar to the standard backup CBF method in~\eqref{def:C_BI}.
To this end, we need to take into account how the state uncertainty affects the backup flow.
Consider the {\em closed-loop}
backup flows for the true and estimated system
\begin{subequations} \label{eq:executed_flow_all}
\begin{flalign}
    \tfrac{\partial}{\partial \tau}\nspace{-1}{\boldsymbol{\phi}_{\rm b}}(\tau,\mathbf{x}) \! &= \! \mathbf{f}(\phitruefullbE) \!+\! \mathbf{g}(\phitruefullbE)\mathbf{k}_{\rm b}(\phihatfullbE), \nspace{50}\label{eq:executedTruth} && \raisetag{2\baselineskip}\\
    \tfrac{\partial}{\partial \tau}\nspace{-1}{\skew{-4}\hat{\boldsymbol{\phi}_{\rm b}}}(\tau,\xhat) \! &= 
    \begin{aligned}[t]
    \! \mathbf{f}(\phihatfullbE) \! &+\!  \mathbf{g}(\phihatfullbE)\mathbf{k}_{\rm b}(\phihatfullbE) \\
    & + \! \mathbf{r}(t \! + \! \tau,\phihatfullbE,\y(t\!+\!\tau)), \label{eq:executedEst} 
    && \raisetag{1.5\baselineskip} 
    \end{aligned} \\
    {\boldsymbol{\phi}}_{\rm b}(0,\x) &= \x, \nspace{4} \skew{5}\hat{\boldsymbol{\phi}}_{\rm b}(0,\xhat) = \xhat, 
\end{flalign}
\end{subequations}
for all ${\tau \in [0,T]}$ where ${T > 0}$ is a fixed integration horizon; cf.~\eqref{eq: nomFlow}.
This system of coupled differential equations represents the forward evolution of \eqref{eq:true-CLsystem} and \eqref{eq:est-CLsystem} under the backup controller ${\ub}$ with output feedback.
Based on the assumptions on systems \eqref{eq:true-CLsystem} and \eqref{eq:est-CLsystem}, for any ${\x \in \mathcal{X}}$ and ${\xhat \in \Xhat}$ there exist unique solutions ${\boldsymbol{\phi}_{\rm b} : \R_{\geq 0} \times \mathcal{X} \rightarrow \mathcal{X}}$ to \eqref{eq:executedTruth} and ${\skew{5}\hat{\boldsymbol{\phi}}_{\rm b} : \R_{\geq 0} \times \Xhat \rightarrow \Xhat}$ to \eqref{eq:executedEst}.
However, both of these solutions are unknown.
The flow $\phitruefullbE$ is unknown because it depends on the true state $\x$, whereas $\phihatfullbE$ is unknown because it depends on future measurements ${\mathbf{y}(t + \tau)}$ appearing in the correction term in \eqref{eq:executedEst}.
Therefore, these backup flows cannot be used directly in control design, but they will help us establish recursive safety and feasibility guarantees.

To expand the backup set and design a safe controller, consider two coupled differential equations which capture the {\em open-loop} backup flows for true and estimated states
\begin{subequations} \label{eq: flows_computed}
\begin{flalign} 
     \tfrac{\partial}{\partial \tau}{\boldsymbol{\phi}^{o}_{\rm b}}(\tau,\mathbf{x})\nspace{-2} &= \nspace{-2} \mathbf{f}(\phitruefullb) \!+\! \mathbf{g}(\phitruefullb)\mathbf{k}_{\rm b}(\phihatfullb),\nspace{50}  \label{eq: computed_flow_true}&& \raisetag{2\baselineskip} \\
     \tfrac{\partial}{\partial \tau}{\skew{-3}\hat{\boldsymbol{\phi}^{o}_{\rm b}}}(\tau,\xhat) \nspace{-2} &= \nspace{-2} \underbrace{\mathbf{f}(\phihatfullb) \!+\! \mathbf{g}(\phihatfullb)\mathbf{k}_{\rm b}(\phihatfullb)}_{\triangleq \nspace{2} \mathbf{f}_{\rm cl}(\phihatfullb)},  \label{eq: computed_flow_est}&& \raisetag{2\baselineskip} \\
    \boldsymbol{\phi}^{o}_{\rm b}(0,\mathbf{x}) &= \x, \quad \skew{5}\hat{\boldsymbol{\phi}}^{o}_{\rm b}(0,\xhat) = \xhat,
\end{flalign}
\end{subequations}
for all ${\tau \in [0,T]}$ with ${T > 0}$.
We refer to this system as ``open-loop'' because the estimate in \eqref{eq: computed_flow_est} is not updated via a correction or innovation term like ${\mathbf{r}(t,\xhat,\y)}$ in \eqref{eq:est-CLsystem}.
The solution to \eqref{eq: computed_flow_est}, denoted ${\phihatfullb}$, is the estimated open-loop backup flow starting at the current state estimate, ${\xhat}$. 
Since there is no innovation term, ${\phihatfullb}$ can now be computed using the known state estimate ${\xhat}$. 
Therefore, the flow ${\phihatfullb}$ can be used
for feedback at run time.
The solution to~\eqref{eq: computed_flow_true}, denoted ${\phitruefullb}$, is the true open-loop backup flow.
It is the evolution of the true state, ${\x}$, propagated with the backup controller using the \textit{estimated} open-loop backup flow, $\mathbf{k}_{\rm b}(\phihatfullb)$.
The flow ${\phitruefullb}$ is unknown because it depends on the true state $\x$.
Thus, it will not be used by the control design directly, but it will help us establish formal safety guarantees.

Having discussed the backup flows, now we expand the backup set into a larger set used for safe control design.
Consider a set defined entirely by the open-loop estimated backup flow, ${\phihatfullb}$,
\begin{align} \label{eq: CI_hatdef}
    \Cihat(t) \triangleq \left\{ \xhat \in \widehat{\mathcal{X}} \,\middle|\, 
    \begin{array}{c}
    h(\phihatfullb) \geq \epsilon_\tau, \forall \nspace{1} \tau \in [0,T], \\
    h_{\rm b}(\phihatfullbT) \geq \epsilon_{\rm b} \\
    \end{array}
    \right\},
\end{align}
with additional tightening terms ${\epsilon_\tau, \epsilon_{\rm b} \!\in\! \R_{\geq 0}}$. These terms can be selected such that conditions on the true state may be inferred from conditions on the estimated state. Since these tightening terms may depend on the integration time variable ${\tau}$ and on the global time ${t}$, ${\Cihat(t)}$ is generalized to a time-varying set in \eqref{eq: CI_hatdef}.
In the following sections, we present two distinct techniques to select these tightening terms, and discuss the implications of each method.
\section{Method 1: Open-loop Safety} \label{sec:method1}

\begin{figure}
    \centering
    \includegraphics[width=1\linewidth]{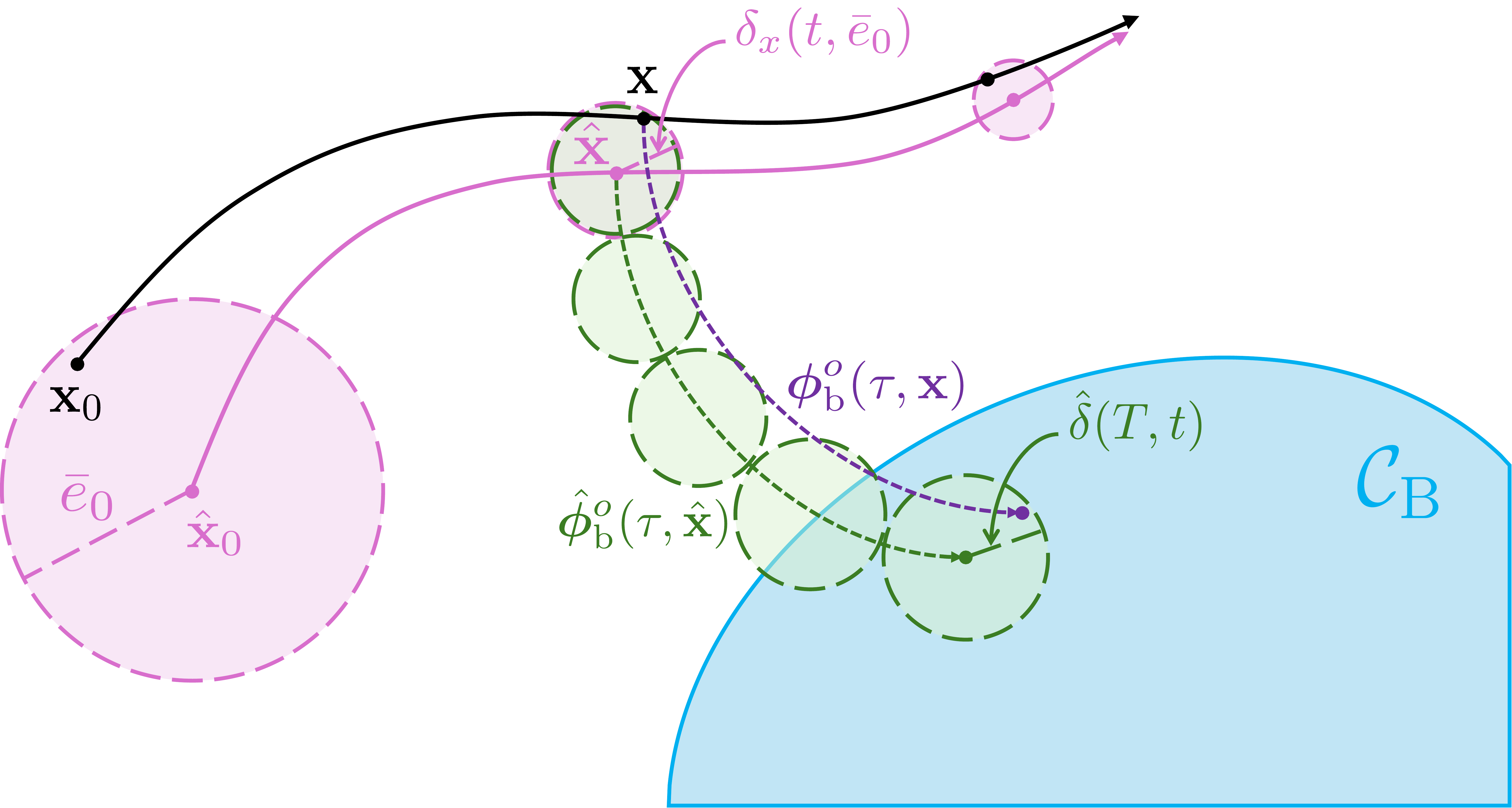}
    \vspace{-.5cm}
    \caption{Sketch of open-loop safety method for the presented output feedback backup CBF approach. Safety of the true state (solid black) is guaranteed if the tube centered on the open-loop estimate (green) contains the open-loop true flow (purple), and can safely reach ${\Cb}$.}
    \label{fig:method1_ch6_overview}
    \vskip -4mm
\end{figure}

The first method to select the tightening terms ${\epsilon_{\tau}}$ and ${\epsilon_{\rm b}}$ is illustrated in \Cref{fig:method1_ch6_overview}. 
In this {\em open-loop} method, the tightening terms are selected such that a tube centered on the estimated open-loop flow $\phihatfullb$ (depicted in green) satisfies the safety constraint (i.e., remains in $\Cs$) and the final cross-section of the tube satisfies the backup set constraint by landing in $\Cb$. We show that if this tube contains the \textit{true} open-loop flow (depicted in purple), output feedback safety of the true state can be maintained.

\subsection{Existence of Feasible Safe Controllers}

Consider a set ${\Ci \subseteq \Cs}$ defined by the true open-loop backup flow, ${\phitruefullb}$, with
\begin{align} \label{eq:CI_def_ch6}
    \Ci \triangleq \left\{ \mathbf{x} \in \mathcal{X} \,\middle|\, 
    \begin{array}{c}
    h(\phitruefullb) \geq 0, \forall \nspace{1} \tau \in [0,T], \\
    h_{\rm b}(\phitruefullbT) \geq 0 \\
    \end{array}
    \right\}.
\end{align}
This set describes all true states which can safely reach the backup set ${\Cb}$ in finite time, under the dynamics in \eqref{eq: computed_flow_true}. 
We can immediately present a powerful result for the set ${\Ci}$, despite it being unknown.
Namely, if the true initial state $\x_0$ is in $\Ci$, then the system~\eqref{eq:true-CLsystem} with the backup controller is safe w.r.t.~$\Cs$ as in Definition~\ref{def:outputSafety}.
\begin{theorem}[Existence of a Feasible and Safe Controller in $\Ci$] \label{thm: CI_controlledinv_kinda}
    If Assumptions~\ref{ass: eb_inCB} and~\ref{ass: robust_inv_inputdist} are satisfied, then for any ${\x_0 \in \Ci}$, there exists a controller ${\mathbf{k}: \Xhat \rightarrow \mathcal{U}}$, ${{\mathbf{u}=\mathbf{k}(\xhat)}}$ for system~\eqref{eq:true-CLsystem}
    such that ${\x(t) \in \Cs}$ for all ${t \geq 0}$.
    \begin{proof}
        By the definition of ${\Ci}$ in \eqref{eq:CI_def_ch6}
        \begin{align}
            \x_0 \in \Ci \implies \boldsymbol{\phi}^{o}_{\rm b}(\tau,\x_0) \in \Cs, \nspace{2} \forall \nspace{1}  \tau \in [0,T]. 
        \end{align}
        For ${\tau > T}$, 
        available measurements are processed
        such that \Cref{ass: eb_inCB} holds and the system evolves with \eqref{eq:executedTruth}, implying $\boldsymbol{\phi}_{\rm b}(0, \boldsymbol{\phi}^{o}_{\rm b}(T,\x_0)) = \boldsymbol{\phi}^{o}_{\rm b}(T,\x_0) \in \Cb$ by the definition of ${\Ci}$. Then, by \Cref{lemma: cb_fwd_inv}, ${\boldsymbol{\phi}_{\rm b}(\vartheta, \boldsymbol{\phi}^{o}_{\rm b}(T,\x_0)) \in \Cb}$ for all ${\vartheta \geq 0}$. The backup set is designed such that ${\Cb \subseteq \Cs}$, and thus the controller ${\ub}$ guarantees that ${\x(t) \in \Cs}$ for all ${t \geq 0}$ for any ${\x_0 \in \Ci}$.
    \end{proof}
\end{theorem}
We now derive tightening terms such that conditions on the true state may be inferred from conditions on the estimated state. The following flow deviation vector will be useful for developing such tightening terms:
\begin{align*}
    \Delta \phi (\tau,\x,\xhat) \triangleq \phitruefullb - \phihatfullb.
\end{align*}
For brevity, we drop the arguments of ${\Delta \phi}$. Though the vector ${\Delta \phi}$ is unknown, if this term is bounded, the tightening terms may be constructed by accounting for the worst-case perturbations to the barrier functions $h$ and $h_{\rm b}$, over the possible realizations of ${\Delta \phi}$. We formalize this with the following lemma.

\begin{lemma}[Conditions on $\Cihat(t)$ Imply Safety on $\Ci$] \label{lemma: pseudoSubsetCi}
    Let ${\hat{\delta}(\tau,t)}$ be a norm bound on the deviation between ${\phihatfullb}$ and ${\phitruefullb}$, given by \eqref{eq: flows_computed}, at time ${\tau \in [0,T]}$ and global time ${t}$, for all ${\mathbf{x} \in \mathcal{X}}$ and ${\xhat \in \Xhat}$:
    \begin{equation} \label{eq:open-loopdeltaBound}
        \big\| \Delta \phi \big\| = \big\|
        {\phitruefullb - \phihatfullb}\big\|
        \leq \hat{\delta}(\tau,t).
    \end{equation}
    If the tightening terms defining $\Cihat(t)$ in \eqref{eq: CI_hatdef} satisfy
    \begin{subequations} \label{eq:sup_epsilon_1}
        \begin{flalign}
           \!& \epsilon_\tau(\tau,t)  \! \geq  \!\!\!\!\sup_{\norm{{\Delta \phi}}\leq\hat{\delta}(\tau,t)} \!\!\!\! h(\phihatfullb) \!- \! h( \phihatfullb \! +  \!{\Delta \phi}), \label{eq:supremum_epsilon_tau_1} \\
           \!& \epsilon_{\rm b}(T,t)  \! \geq  \!\!\!\!\sup_{\norm{{\Delta \phi}}\leq\hat{\delta}(T,t)} \!\!\!\! h_{\rm b}(\phihatfullbT) \! - \! h_{\rm b}( \phihatfullbT \! +  \!{\Delta \phi}), \label{eq:supremum_epsilon_b_1} \nspace{50} \raisetag{1.5\baselineskip}
        \end{flalign}
    \end{subequations}
    for all ${\tau \in [0,T]}$ and ${t \geq 0}$, 
    then ${\xhat \in \Cihat(t) \implies \x \in \Ci}$.
    \begin{proof}
        Though the tightening terms in \eqref{eq:sup_epsilon_1} have an explicit dependence on $\tau$ and $t$, we suppress this dependence for the remainder of the paper.
        It is clear that if ${\epsilon_\tau \geq h(\phihatfullb) - h\big( \phihatfullb + {\Delta \phi}\big)}$ then
        \begin{align}
            h(\phitruefullb) 
            \!&=\! h(\phihatfullb) \!-\! ( h(\phihatfullb)  \!-\!h( \phihatfullb \!+\! {\Delta \phi}) ) \notag \\
            &\geq h(\phihatfullb) - \epsilon_\tau. \label{eq:inequality_intermedLem3}
        \end{align}
        Membership of $\xhat$ to $\Cihat(t)$ implies ${h(\phihatfullb) \geq \epsilon_{\tau}}$ so that from \eqref{eq:inequality_intermedLem3}, ${h(\phitruefullb) \geq 0}$ for any ${\xhat \in \Cihat(t)}$.
        Similarly, for any ${\xhat \in \Cihat(t)}$, the open-loop backup trajectory starting at ${\xhat}$ satisfies ${h_{\rm b}(\phihatfullbT) \geq \epsilon_{\rm b}}$ and by the same logic,
        \begin{align*}
            h_{\rm b}(\phitruefullbT)
            &\geq h_{\rm b}(\phihatfullbT) - \epsilon_{\rm b} \geq 0,
        \end{align*}
        if ${\epsilon_{\rm b} \geq h_{\rm b}(\phihatfullbT) - h_{\rm b}\big( \phihatfullbT + {\Delta \phi}\big)}$. Thus, all functions defining $\Ci$ are positive, so ${\mathbf{x} \in \Ci}$ for all ${\xhat \in \Cihat(t)}$.
    \end{proof}
\end{lemma}

\Cref{lemma: pseudoSubsetCi} allows us to establish the existence of a safe and feasible controller within the known set, $\Cihat(t)$.

\begin{theorem}[Existence of a Feasible and Safe Controller in $\Cihat(t)$] \label{thm:feasible_Cihat}
    If Assumptions~\ref{ass: eb_inCB} and~\ref{ass: robust_inv_inputdist} are satisfied
    and if ${\epsilon_{\tau}}$ and ${\epsilon_{\rm b}}$ satisfy \eqref{eq:sup_epsilon_1} for all ${\tau \in [0,T]}$ and ${t \geq 0}$,
    then for any ${\xhat_0 \in \Cihat(0)}$, there exists a controller ${{\mathbf{k}:\Xhat \rightarrow \mathcal{U}}}$, ${{\mathbf{u}=\mathbf{k}(\xhat)}}$ for system~\eqref{eq:true-CLsystem} such that ${\x(t) \in \Cs}$ for all ${t \geq 0}$.
    \begin{proof}
        By \Cref{lemma: pseudoSubsetCi} ${{\xhat_0 \!\in\! \Cihat(0) \!\!\implies\!\! \x_0 \!\in\! \Ci}}$ and by \Cref{thm: CI_controlledinv_kinda}, the controller ${\ub}$ ensures that ${{\x(t) \in \Cs, \nspace{1} \forall \nspace{1} t \geq 0}}$.
    \end{proof}
\end{theorem}

\subsection{Tightening Terms and Open-Loop Flow Bounds}

We now discuss the computation of the tightening terms $\epsilon_\tau , \epsilon_{\rm b}$ and the open-loop flow bound $\hat{\delta}(\tau,t)$.

\begin{remark}[Computing Tightening Terms]\label{rem:tighteningTerms}
For linear barrier functions, the supremum in  $\eqref{eq:sup_epsilon_1}$ may be obtained exactly, yielding tight ${\epsilon_\tau}$ and $\epsilon_{\rm b}$ terms. For quadratic barriers, the supremum over the ball of flow deviations is in the form of a \textit{trust region problem} from classical numerical optimization \cite[Ch. 4]{nocedal_numerical_2006}. Such a barrier is used effectively in \Cref{ex:db_int}. For differentiable convex barrier functions, $\eqref{eq:sup_epsilon_1}$ may be satisfied\footnote{Using the first-order condition of a differentiable convex function, ${-\nabla h (\phihatfullb)^\top \Delta\phi \geq h(\phihatfullb) - h(\phihatfullb + \Delta \phi)}$ and properties of the dual Euclidean norm.} by 
\begin{subequations}
    \begin{align*}
       \epsilon_\tau
       &\triangleq \big\|{\nabla h (\phihatfullb)}\big\|\hat{\delta}(\tau,t), \\
       \epsilon_{\rm b}
       &\triangleq \big\|{\nabla h_{\rm b} (\phihatfullbT)}\big\|\hat{\delta}(T,t).
    \end{align*}
\end{subequations}
For more complex barriers, $\eqref{eq:sup_epsilon_1}$ may be satisfied by utilizing the Lipschitz constants of $h$ and $h_{\rm b}$, yielding simpler but more conservative tightening terms.
These terms take the form ${\epsilon_{\tau} \triangleq \mathcal{L}_{h} \hat{\delta}(\tau,t)}$ and ${\epsilon_{\rm b} \triangleq \mathcal{L}_{h_{\rm b}} \hat{\delta}(T,t)}$ where ${\mathcal{L}_{h}}$ and ${\mathcal{L}_{h_{\rm b}}}$ are the Lipschitz constants of ${h}$ and ${h_{\rm b}}$, respectively. This last form of tightening term has been used effectively in related works \cite{vanWijk_DRbCBF24,vanwijk2025_dobcbf},\cite[Lemma 1]{gurriet_scalable}, but \Cref{lemma: pseudoSubsetCi} is a generalization of this form, allowing for tighter bounds when the structure of the barrier may be leveraged, or when numerical optimization tools can be used.
\end{remark}

To effectively compute the tightening terms, the flow bound $\hat{\delta}(\tau,t)$ must be obtained between the estimated open-loop backup flow and the true open-loop backup flow.

\begin{lemma}[General Open-loop Flow Bounds] \label{lemma: flowbound_method1_ch6}
For systems \eqref{eq:true-CLsystem} and \eqref{eq:est-CLsystem} let ${\mathbf{f}}$ and ${\mathbf{g}}$ be locally Lipschitz on ${\mathcal{X}}$ and ${\Xhat}$ with respective Lipschitz constants ${\mathcal{L}_{\mathbf{f}}}$ and ${\mathcal{L}_{\mathbf{g}}}$. If ${\norm{\ub(\xhat)} \leq \bar{u}}$ for all ${\xhat \in \Xhat}$, the following bound satisfies condition \eqref{eq:open-loopdeltaBound}
\begin{align} \label{eq:flowbound}
\hat{\delta}(\tau,t) \triangleq \delta_x(t, \Bar{e}_0){\rm e}^{(\mathcal{L}_{\mathbf{f}} + \mathcal{L}_{\mathbf{g}} \bar{u})\tau},
\end{align}
for ${\tau \in {[0,T]}}$ and a function 
${\delta_x}$ satisfying \Cref{def: error_estimator}.
\begin{proof}
    Recall that ${\Delta \phi(\tau,\mathbf{x}, \xhat) \!\triangleq\! \phitruefullb \!-\! \phihatfullb}$, and for notational simplicity we omit the dependence on $\tau$, $\x$, and $\xhat$. Applying Lipschitz continuity of ${\mathbf{f}}$ and ${\mathbf{g}}$
    \begin{align*}
        \norm{\Delta \phi } \leq \norm{\xhat - \x} \! + \!\! \int^\tau_0 \Big( \mathcal{L}_{\mathbf{f}} + \mathcal{L}_{\mathbf{g}} \|{\ub(\skew{5}\hat{\boldsymbol{\phi}}^{o}_{\rm b}(s,\xhat))}\| \Big) \norm{\Delta \phi } ds.
    \end{align*}
    Since the backup controller is bounded by ${\bar{u}}$, applying the Gr{\"o}nwall--Bellman inequality \cite[Lemma 2.1]{khalil2002nonlinear} and \Cref{def: error_estimator} completes the proof
    \begin{align*}
        \norm{\Delta \phi } \leq \norm{\xhat - \x}{\rm e}^{(\mathcal{L}_{\mathbf{f}} + \mathcal{L}_{\mathbf{g}} \bar{u})\tau} \leq \delta_x(t, \Bar{e}_0)\nspace{1}{\rm e}^{(\mathcal{L}_{\mathbf{f}} + \mathcal{L}_{\mathbf{g}} \bar{u})\tau}.
        &\qedhere
    \end{align*}
\end{proof}
\end{lemma}
\begin{remark}[Practicality of Bounds] \label{rem: method1-bound-bad}
    Because this bound holds for a very broad class of systems, it may be impractical for systems where ${\mathcal{L}_{\mathbf{f}}}$ and ${\mathcal{L}_{\mathbf{g}}}$ assume large values. For systems with constant ${\mathbf{g}}$, ${\mathcal{L}_{\mathbf{g}} = 0}$, and the bound will reduce to ${\hat{\delta}(\tau,t) = \delta_x(t, \Bar{e}_0){\rm e}^{\mathcal{L}_{\mathbf{f}}\tau}}$, which can be used for short time horizons, ${T}$. For linear systems, this bound can be made tight (see \Cref{sec:LinearCH6}). Alternatively, one-sided Lipschitz constants may replace ${\mathcal{L}_{\mathbf{f}}}$ and ${\mathcal{L}_{\mathbf{g}}}$. The upper bounds
    on the logarithmic norm (or matrix measure) of $\mathbf{f}$ and $\mathbf{g}$ are typically tighter than two-sided Lipschitz bounds \cite{Sontag2010}.
\end{remark}

\subsection{Output Feedback Safety-Critical Control}
 
Finally, since we have shown that membership of the \textit{estimated} state to ${\Cihat(t)}$ implies membership of the true state to ${\Ci \!\subseteq\! \Cs}$ (Lemma~\ref{lemma: pseudoSubsetCi}), and that there exists at least one safe
controller in ${\Cihat(t)}$ (Theorem~\ref{thm:feasible_Cihat}), we may derive forward invariance conditions on ${\Cihat(t)}$. By \eqref{eq: CI_hatdef}, this requires 
\begin{align} \label{eq: nagumo_cihat}
    \begin{gathered}
    \dot{h}(\phihatfullb, \mathbf{u}) - \dot{\epsilon}_\tau \ge - \alpha \big( h(\phihatfullb) - \epsilon_\tau \big), \\
    \dot{h}_{\rm b}(\phihatfullbT, \mathbf{u}) - \dot{\epsilon}_{\rm b} \ge -\alpha_{\rm b} \big( h_{\rm b}(\phihatfullbT) - \epsilon_{\rm b} \big),
    \end{gathered}
\end{align}
for extended class-${\mathcal{K}_{\infty}}$ functions ${\alpha,\alpha_{\rm b}}$ and for all ${\tau \in [0,T]}$. The total derivatives of $h$ and $h_{\rm b}$ are given by
\begin{align*}
    \begin{gathered}
        \dot{h}(\phihatfullb, \mathbf{u}) = \nabla h(\phihatfullb) \stmhat \xhatDot,\\
        \dot{h}_{\rm b}(\phihatfullbT, \mathbf{u}) = \nabla h_{\rm b}(\phihatfullbT) \stmhatT \xhatDot.
    \end{gathered}
\end{align*}
Here ${\stmhat \nspace{-1}\!\triangleq\!\nspace{-1} {\partial \phihatfullb}\nspace{-1}/\nspace{-1}{\partial \xhat}}$ is the sensitivity of the estimated backup flow to ${\xhat}$, and ${\xhatDot \!=\! \mathbf{f}(\hat{\x}) \!+\! \mathbf{g}(\hat{\x})\mathbf{u} \!+\! \mathbf{r}(t,\xhat,\y)}$. The sensitivity matrix evolves with
\begin{align}
    \tfrac{\partial }{\partial  {\tau}}\stmhat \!=\! \jac(\phihatfullb) \stmhat, \nspace{10} \widehat{\mathbf{\Phi}}_{\rm b}(0,\xhat)\!= \!\mathbf{I}_{n},
\end{align}
where
${\jac(\xhat) \triangleq \frac{\partial}{\partial \xhat}\big( 
\mathbf{f}_{\rm cl}(\xhat)
\big )}$
is the Jacobian of the closed-loop system \eqref{eq: computed_flow_est}. Because no innovation term is present, system \eqref{eq: computed_flow_est} can be viewed as a state feedback system, and thus the closed-loop Jacobian is straightforward to calculate using similar methods as in \Cref{sec:bCBF}.

We now establish the main result -- controllers satisfying these conditions ensure the safety of the true system \eqref{eq:true-CLsystem}, using only information about the estimated system. 

\begin{theorem}[Forward Invariance of $\Cihat(t)$] \label{thm: fwd_inv_Cihat}
    Any locally Lipschitz controller ${{\mathbf{k} \!:\! \Xhat \!\times\! \mathbb{R}_{\geq 0} \!\to\! \U}}$, ${{\mathbf{u} \!=\! \mathbf{k}(\xhat,t)}}$ satisfying 
    \begin{subequations} \label{eq: mainthmConstraints_ch6}
        \begin{align} 
            \!&\begin{aligned}
                \nabla h(\phihatfullb) \stmhat \big ( \mathbf{f}(\hat{\x}) + \mathbf{g}(\hat{\x})\mathbf{u} \big ) \ge \\
                - \alpha \big( h(\phihatfullb) - \epsilon_\tau \big) + \dot{\epsilon}_\tau + \hat{\rho}, \label{eq: mainthmConstraints_cont_ch6}
             \end{aligned} \\ 
            \!&\begin{aligned}
                \nabla h_{\rm b}(\phihatfullbT) \stmhatT \big ( \mathbf{f}(\hat{\x}) + \mathbf{g}(\hat{\x})\mathbf{u} \big )  \ge \\
                -\alpha_{\rm b} \big( h_{\rm b}(\phihatfullbT) - \epsilon_{\rm b} \big) +  \dot{\epsilon}_{\rm b} + \hat{\rho}_{\rm b}, 
                \label{eq: mainthmConstraints_reach_ch6}
             \end{aligned} 
        \end{align}
    \end{subequations}
    for all ${\tau \in [0,T]}$, ${t\geq0}$ and ${\xhat \in \Cihat(t)}$ with robustness terms
    \begin{align} \label{eq:robust_terms_mainthm}
    \begin{gathered}
        \hat{\rho} \geq - \nabla h(\phihatfullb) \stmhat \mathbf{r}(t,\xhat,\y), \\
        \hat{\rho}_{\rm b} \geq - \nabla h_{\rm b}(\phihatfullbT) \stmhatT \mathbf{r}(t,\xhat,\y),
    \end{gathered}
    \end{align}
    renders ${\Cihat(t)}$ forward invariant for \eqref{eq:est-CLsystem}, which implies the controller enforces ${\x(t) \!\in\! \Ci\! \subseteq \!\Cs}$, ${\forall \nspace{2} t \!\geq\! 0}$ for \eqref{eq:true-CLsystem}.
    \begin{proof}
        By applying \Cref{thm: cbf} to system \eqref{eq:est-CLsystem}, satisfaction of \eqref{eq: mainthmConstraints_ch6} yields forward invariance of ${\Cihat(t)}$. By \Cref{lemma: pseudoSubsetCi}, membership of ${\xhat}$ to ${\Cihat(t)}$ implies membership of ${\x}$ to ${\Ci \subseteq \Cs}$, thus guaranteeing safety of ${\x}$ w.r.t. ${\Cs}$.
    \end{proof}
\end{theorem}

\begin{remark}[Utility of Robustification]
    Note that when the inequalities in \eqref{eq:robust_terms_mainthm} are equalities, the invariance conditions from \eqref{eq: nagumo_cihat} are recovered exactly. 
    Though the innovation term ${\mathbf{r}(t,\xhat,\mathbf{y})}$ is known, it was found that robustifying against the innovation (i.e., over-approximating this term in \eqref{eq:robust_terms_mainthm}) can help the performance of the safety filter at the boundary of $\Cihat(t)$ by preventing excessive chattering caused by the innovation.
\end{remark}
\begin{corollary}[Robustness Terms for Linear Correction]
    For an estimator with a linear correction term as in \Cref{ex:linear_correction}, the condition \eqref{eq:robust_terms_mainthm} can be satisfied with robustness terms
    \begin{flalign} \label{eq:linear_robustness}
    \begin{aligned}
     \hat{\rho} \!&=\! \norm{ \nabla h(\phihatfullb) \stmhat \mathbf{L}(t)}\!\Big(\mathcal{L}_{\mathbf{z}} \delta_x(t,\Bar{e}_0) \!+\! \Bar{v} \Big), \nspace{50} \raisetag{1.75\baselineskip} \\
    \hat{\rho}_{\rm b} \!&=\! \norm{ \nabla h_{\rm b}(\phihatfullbT) \stmhatT \mathbf{L}(t)}\!\Big(\mathcal{L}_{\mathbf{z}} \delta_x(t,\Bar{e}_0) \!+\! \Bar{v}\Big),
    \end{aligned}
    \end{flalign}
    where ${\mathcal{L}_{\mathbf{z}}}$ is the Lipschitz constant of the measurement function $\mathbf{z}$ in \eqref{eq:measFun} and $\Bar{v}$ is the measurement noise bound.
    \begin{proof}
        First substitute ${\mathbf{L}(t)\big(\mathbf{y} - \mathbf{z}(\hat{\x})\big)}$ from \eqref{eq:est_Luen} in for ${\mathbf{r}(t,\xhat,\y)}$ in \eqref{eq:robust_terms_mainthm}. By Lipschitz continuity of the measurement function, ${\norm{\mathbf{y} - \mathbf{z}(\hat{\x})} \leq \mathcal{L}_{\mathbf{z}} \delta_x(t,\Bar{e}_0) + \Bar{v}}$. Thus by the Cauchy--Schwarz inequality, the conditions in \eqref{eq:linear_robustness} imply the satisfaction of \eqref{eq:robust_terms_mainthm}.
    \end{proof}
\end{corollary}

We can now use \Cref{thm: fwd_inv_Cihat} to develop a novel safe controller via \textit{Output Feedback
Backup CBFs (O-bCBFs)}:
\begin{align*} 
    \mathbf{k}^\star(\xhat,t) = \underset{\mathbf{u} \in \mathcal{U}}{\text{argmin}} \mkern9mu &
    \left\Vert \mathbf{k}_{\rm p}(\xhat,t)-\mathbf{u}\right\Vert^{2} \quad \tag{O-bCBF-QP} \label{eq:e-bcbf-qp} \\
    \text{s.t.  } 
    & \eqref{eq: mainthmConstraints_cont_ch6}, \ \eqref{eq: mainthmConstraints_reach_ch6},
\end{align*}
for all ${\tau \in [0,T]}$ and $t \geq 0$, which filters the primary controller ${\mathbf{k}_{\rm p} : \Xhat \times \mathbb{R}_{\geq 0} \rightarrow \mathcal{U}}$. 
Because ${\Cihat(t)}$ is not necessarily controlled invariant, the \eqref{eq:e-bcbf-qp} may technically experience infeasibility issues.
However, by \Cref{thm:feasible_Cihat}, the backup controller ${\ub(\xhat) \in \U}$ can guarantee the safety of $\x$ for any ${\xhat \in \Cihat(t)}$. Thus, if the \eqref{eq:e-bcbf-qp} becomes infeasible, one may use the backup controller to guarantee output feedback safety\footnote{According to \Cref{def:outputSafety} for ${\widehat{\mathcal{X}}_{0} \triangleq \Cihat(0)}$, ${\X_0 \triangleq \widehat{\mathcal{X}}_{0} \oplus \mathcal{B}_{\Bar{e}_0}}$ where ${\mathcal{B}_{\Bar{e}_0} \triangleq \{ \mathbf{e}_x(0) \in \mathbb R^n : \norm{\mathbf{e}_x(0)} \leq \Bar{e}_{0} \}}$.} under input bounds.
A controller realizing this idea takes the form
\begin{align} \label{eq:switched_safeControl}
    \mathbf{k}_{\rm safe}(\xhat,t) = \mathds{1}_{\mathcal{F}}( \xhat,t)\mathbf{k}^\star(\xhat,t) + \big(1 \!-\! \mathds{1}_{\mathcal{F}}(\xhat,t)\big)\mathbf{k}_{\rm b}(\xhat).
\end{align}
Here, $\mathds{1}_{\mathcal{F}}$ is the indicator function defined such that ${\mathds{1}_{\mathcal{F}}(\xhat,t) = {1}}$ if there exists a ${ \boldsymbol{u} \in \mathcal{U}}$ such that \eqref{eq: mainthmConstraints_cont_ch6} and \eqref{eq: mainthmConstraints_reach_ch6} hold, and $0$ otherwise. 
A smooth switching approach as in \cite{rabiee2025soft} could also be used. We note that in the simulation examples presented in \Cref{sec:examples}, infeasibility of the \eqref{eq:e-bcbf-qp} was not observed.

\subsection{Linear Systems} \label{sec:LinearCH6}

In order to glean additional insight about safety-critical output feedback control, we examine the special case where the dynamics and the measurement function are linear.
We provide more details on how to compute tighter flow bounds $\hat{\delta}(\tau,t)$, determine the estimator error bound $\delta_x(t,\Bar{e}_0)$, and design backup sets and backup controllers.

Consider the special case of \eqref{eq:truth_eqtns},~\eqref{eq:est}, and~\eqref{eq:feedbackController} where the dynamics are linear time-invariant, such that the closed-loop output feedback system in \eqref{eq:cl_outputFeedback} becomes
\begin{subequations} \label{eq:CL_eqtns_linear}
\begin{align}
    &\dot{\mathbf{x}} = \A\mathbf{x} + \B\mathbf{u}, \label{eq: truth_linear} \\
    &\mathbf{y} = \C \mathbf{x} + \mathbf{v}(t), \\
    &\dot{\hat{\x}} = \A\hat{\x} + \B\mathbf{u} + \mathbf{r}(t,\xhat,\y),  \label{eq: est_linear} \\
    &\mathbf{u} = \mathbf{k}(\xhat,t), \\
    &\x(0) = \x_0, \nspace{4} \xhat(0) = \xhat_0.
\end{align}
\end{subequations}
As before ${\mathbf{x} \!\in\! \mathcal{X} \!\subset \!\mathbb{R}^n}$, ${\xhat \!\in\! \Xhat \!\subset\! \mathbb{R}^n}$, ${\mathbf{u} \!\in\! \mathcal{U} \!\subseteq\! \mathbb{R}^m}$, ${\y \!\in\!  \mathcal{Y} \!\subseteq\! \mathbb{R}^{y}}$.
The dynamics are governed by ${\A \!\in\! \mathbb{R}^{n\times n}}$ and ${\B \!\in\! \mathbb{R}^{n\times m}}$, ${\C \!\in\! \mathbb{R}^{y \times n}}$ is the output matrix, 
and the feedback controller ${\mathbf{k} : \Xhat \times \mathbb{R}_{\geq 0} \rightarrow \mathcal{U}}$ is assumed to be Lipschitz continuous in ${\xhat}$ and piecewise continuous in ${t}$. The pair (${\A, \C}$) is assumed to be observable.
For this class of systems, the true and estimated open-loop backup flows in \eqref{eq: flows_computed} reduce to the following system of ODEs
\begin{subequations} \label{eq: flows_computed_linear}
    \begin{align}
        \tfrac{\partial }{\partial  {\tau}}{\boldsymbol{\phi}^{o}_{\rm b}}(\tau,\mathbf{x}) &= \mathbf{A}\phitruefullb + \mathbf{B}\mathbf{k}_{\rm b}(\phihatfullb),  \label{eq: computedbackup_truth_linear} \\
         \tfrac{\partial }{\partial  {\tau}}\skew{-3}\hat{\boldsymbol{\phi}^{o}_{\rm b}}(\tau,\xhat) &=  \mathbf{A}\phihatfullb +  \mathbf{B}\mathbf{k}_{\rm b}(\phihatfullb), \label{eq: computedbackup_est_linear} \\
         \boldsymbol{\phi}^{o}_{\rm b}(0,\mathbf{x}) &= \x,   \quad \skew{5}\hat{\boldsymbol{\phi}}^{o}_{\rm b}(0,\xhat) = \xhat.
    \end{align}
\end{subequations}

We show that for linear systems, the backup flow bound presented in \Cref{lemma: flowbound_method1_ch6} can be improved significantly.
\begin{lemma}[Linear System Flow Bounds]
For systems \eqref{eq: computedbackup_truth_linear} and \eqref{eq: computedbackup_est_linear}, if a time-varying function 
${\delta_x}$
satisfies \Cref{def: error_estimator}, then the bound
\begin{align} \label{eq: linearBound_method1}
    \hat{\delta}(\tau, t) \triangleq \delta_x(t, \Bar{e}_0) \norm{{\rm e}^{\A \tau}},
\end{align}
satisfies \eqref{eq:open-loopdeltaBound} for all ${\tau \in [0, T]}$ and ${t \geq 0}$.
\begin{proof}
    Introducing ${\Delta \phi(\tau,\mathbf{x}, \xhat) \triangleq \phitruefullb - \phihatfullb}$ 
    \begin{align*}
         \tfrac{\partial }{\partial  {\tau}} \Delta \phi(\tau,\mathbf{x},\xhat) = \mathbf{A} \Delta \phi(\tau,\mathbf{x},\xhat),
         \sspace \Delta \phi(0,\mathbf{x},\xhat) = \x - \hat{\x},
    \end{align*}
    which has the solution ${
        \Delta \phi(\tau,\mathbf{x},\xhat) \!=\! {\rm e}^{\mathbf{A}\tau} (\x - \hat{\x})}$.
    Thus, the linear system flow bound is
    \begin{align*}
        \norm{\Delta \phi(\tau,\mathbf{x},\xhat)} 
        \leq \norm{{\rm e}^{\mathbf{A}\tau}}\underbrace{\norm{\x(t) - \xhat(t)}}_{\leq \delta_x(t, \Bar{e}_0)}. &\qedhere
    \end{align*}
\end{proof}
\end{lemma}
Thus, for the open-loop method for linear systems, we can avoid applying Gr{\"o}nwall's inequality, and the bound's dependence on the control input completely disappears. 

For a linear innovation term (see \Cref{ex:linear_correction}) it is also straightforward to derive a bound $\delta_x(t, \Bar{e}_0)$ on the estimation error such that \eqref{eq: est_linear} is an error-bounded estimator according to \Cref{def: error_estimator}.

\begin{proposition}[Analytical Estimation Error Bounds]
\label{prop:linearBoundEstimation}
    For \eqref{eq: truth_linear} and \eqref{eq: est_linear} with a linear correction term as in \Cref{ex:linear_correction}, and a constant observer gain ${\mathbf{L}}$, the estimation error ${\mathbf{e}_x(t) = \x(t) - \xhat(t)}$ is bounded by:
    \begin{flalign} \label{eq:delx_linear}
        \norm{\mathbf{e}_x(t)} \!\leq\! \Bar{e}_0 \norm{{\rm e}^{\bLam t}} \!+\! \Bar{v}\! \int^t_0\!\norm{{\rm e}^{\bLam (t - \vartheta)}\mathbf{L}} \nspace{1} \!d\vartheta  \triangleq  \delta_x(t, \Bar{e}_0), 
    \end{flalign}
    where ${\bLam \triangleq \A - \mathbf{L}\C}$ and ${\norm{\e_x(0)} \leq \Bar{e}_0 \in \mathbb{R}_{>0}}$.
    \begin{proof}
        The error dynamics, ${{\dot{\mathbf{e}}_x}(t) \!=\! (\mathbf{A} \!-\! \mathbf{L} \C) \mathbf{e}_x(t) \!-\! \mathbf{L} \mathbf{v}(t)}$,
        have the solution 
        \begin{align*}
            \mathbf{e}_x(t) = {\rm e}^{\bLam t}\mathbf{e}_x(0) - \int^t_0 {\rm e}^{\bLam (t-\vartheta)} \mathbf{L} \mathbf{v}(\vartheta) \nspace{1} d\vartheta.
        \end{align*}
        Taking norms and recalling ${\norm{\mathbf{v}(t)} \!\leq\! \Bar{v}}$ yields
        \begin{align*}
            \norm{\mathbf{e}_x(t)} \leq \norm{\mathbf{e}_x(0)} \norm{{\rm e}^{\bLam t}} + \Bar{v} \int^t_0\norm{{\rm e}^{\bLam (t - \vartheta)}{\mathbf{L}}} \nspace{1} d\vartheta.
        \end{align*}
        Using ${\norm{\e_x(0)} \leq \Bar{e}_0}$ completes the proof.
    \end{proof}
\end{proposition}
The observer gain ${\mathbf{L}}$ should be chosen such that ${\bLam}$ is Hurwitz.
Note that \eqref{eq:delx_linear} is one of many possible bounds for linear observers, and it is provided for its simplicity.

We now demonstrate that an ellipsoidal level set centered on an equilibrium point inside $\Cs$ may be used as a backup set for \eqref{eq:CL_eqtns_linear}, and we provide conditions for which a simple output feedback backup controller renders the backup set forward invariant, allowing us to avoid making \Cref{ass: robust_inv_inputdist}. Without loss of generality, we take the equilibrium point to be the origin.

\begin{proposition}[Backup Set and Controller Design] \label{prop:linearsystemCBandub}
Consider a backup set given by
\begin{align} \label{eq: backupSetLinear_general}
    \mathcal{C}_{\rm B} \triangleq \{\mathbf{x} \in \mathbb{R}^n : h_{\rm b}(\x) = \gamma - \x^\top \mathbf{P} \x \geq 0  \} \subseteq \Cs,
\end{align}
with ${\gamma \!\in\! \R_{>0}}$ and symmetric positive definite matrix ${\mathbf{P} \!\succ\! 0}$. For systems \eqref{eq: truth_linear} and \eqref{eq: est_linear} satisfying \Cref{ass: eb_inCB}, the output feedback backup controller ${\ub(\xhat) = -\mathbf{K}\xhat}$ renders $\Cb$
forward invariant for gains ${\mathbf{K} \in \R_{\geq 0}^{m \times n}}$ which satisfy
\begin{align} \label{eq:conditionubLinearRobust}
    \lambda_{\rm min}(\mathbf{Q}) \geq 2 \Bar{e}_{\rm b} 
    \sqrt{{\lambda_{\rm min}(\mathbf{P})}/{\gamma}}
    \norm{\mathbf{P} \B \K},
\end{align}
for ${\mathbf{Q} \!=\! \mathbf{Q}^\top \!=\! -\big( \big(\A - \mathbf{B}\K\big)^{\top}\mathbf{P} + \mathbf{P}\big(\A - \mathbf{B}\K\big) \big) \!\succ\! 0}$.
\begin{proof}
    The closed-loop dynamics of \eqref{eq: truth_linear} with ${\mathbf{u} = -\mathbf{K}\xhat}$ are given by ${\dot{\x} = \mathbf{A}\x - \mathbf{B}\mathbf{K}\xhat}$ and thus
    adding zero we have
    \begin{align}
            \dot{\x} 
            &= \underbrace{\big(\A - \mathbf{B}\K\big)}_{\triangleq \nspace{2}\A_{\rm cl}}\x + \B\K\underbrace{(\x - \xhat)}_{\e_x}. \label{eq: cb_cl_2}
    \end{align}
    To prove forward invariance of ${\Cb}$, we expand the time derivative of ${h_{\rm b}}$
    and collect like terms such that
    \begin{align}
        \dot{h}_{\rm b}(\x,\xhat) &= -\x^\top \underbrace{\big(\A_{\rm cl}^\top \mathbf{P} + \mathbf{P} \A_{\rm cl} \big)}_{\triangleq \nspace{2} -\mathbf{Q}} \x - 2 \x^\top \mathbf{P} \B \K \e_x.
    \end{align}
    If ${\mathbf{Q}}$ is positive definite, then ${\x^\top \mathbf{Q} \x}$ will be a positive scalar. Thus, for forward invariance of ${\Cb}$, we require
    \begin{align} \label{eq: intermed1}
        \x^\top \mathbf{Q} \x \geq 2 \x^\top \mathbf{P} \B \K \e_x, \quad \forall \nspace{2} \x \in \partial\Cb.
    \end{align}
    Along $\partial\Cb$ we have\footnote{Given a symmetric matrix $\mathbf{S}$, we let $\lambda_{\rm min}(\mathbf{S})$ and 
    $\lambda_{\rm max}(\mathbf{S})$ denote its minimum and maximum eigenvalues, respectively.} that
    ${\gamma = \x^\top \mathbf{P} \x \geq \lambda_{\rm min}(\mathbf{P}) \norm{\x}^2}$, so
    \begin{align} \label{eq: x_boundary}
        \norm{\x} \leq 
        \sqrt{{\gamma}/{\lambda_{\rm min}(\mathbf{P})}}, \quad \forall \nspace{2} \x \in \partial\Cb.
    \end{align}
    With
    \Cref{ass: eb_inCB} one has
    ${2 \x^\top \mathbf{P} \B \K \e_x \!\leq\! 2 \!\norm{\x}\! \norm{\mathbf{P}\B\K}\! \Bar{e}_{\rm b}}$, and since $\Q$ is symmetric we have ${\x^\top \mathbf{Q} \x \!\geq\! \lambda_{\rm min}(\mathbf{Q}) \norm{\x}^2}$.
    Thus, the condition ${\lambda_{\rm min}(\mathbf{Q}) \norm{\x}^2 \!\geq\! 2 \norm{\x} \norm{\mathbf{P}\B\K} \Bar{e}_{\rm b}}$
    implies satisfaction of \eqref{eq: intermed1}. Simplifying and applying \eqref{eq: x_boundary} completes the proof.
\end{proof}
\end{proposition}
Note that in the case where the controller in \Cref{prop:linearsystemCBandub} is designed to satisfy \Cref{ass: robust_inv_inputdist}, the condition on ${\K}$ is very similar to \eqref{eq:conditionubLinearRobust}, but slightly more conservative. Recognizing that in \Cref{prop:linearsystemCBandub} the virtual disturbance to consider, ${\mathbf{d}_v = \B \K \e_x}$, is bounded as ${\norm{\mathbf{d}_v} \leq \norm{\B}\norm{\K}\Bar{e}_{\rm b}}$,
then to satisfy \Cref{ass: robust_inv_inputdist}, the requirement on the gain matrix is:
\begin{align}
    \lambda_{\rm min}(\mathbf{Q}) \geq 2 \Bar{e}_{\rm b} 
    \sqrt{{\lambda_{\rm min}(\mathbf{P})}/{\gamma}}
    \norm{\mathbf{P}}\norm{\B}\norm{\K}.
\end{align}
\section{Method 2: Closed-loop Safety} \label{sec:method2}

\begin{figure}
    \centering
    \includegraphics[width=1\linewidth]{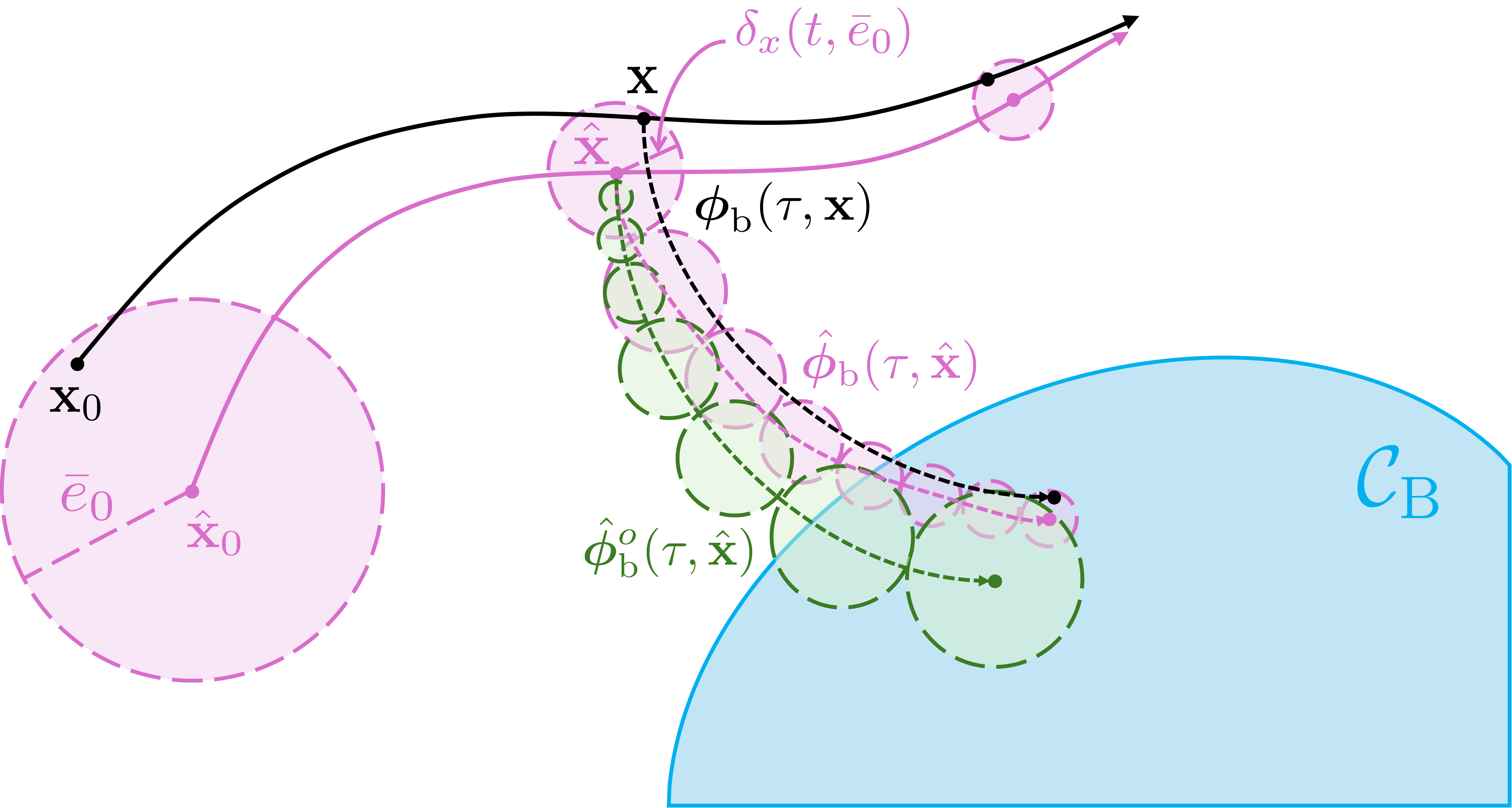}
    \caption{Illustration of closed-loop safety method for the presented output feedback CBF approach.
    The norm ball tube centered on the open-loop estimate (green) contains the closed-loop estimated flow (pink dashed), and the norm ball tube centered on the closed-loop estimated flow (pink) contains the true closed-loop flow (black dashed). Safety of the true state (solid black) is guaranteed if the tube constructed from the sum of the green and pink tubes can safely reach ${\Cb}$.
    }
    \vskip -4mm
    \label{fig:method2_ch6_overview}
\end{figure}

This section describes the second method of selecting the tightening terms ${\epsilon_{\tau}}, {\epsilon_{\rm b}}$ (depicted in \Cref{fig:method2_ch6_overview}) by using information about the \textit{closed-loop} system.
Since the estimated and true closed-loop flows in \eqref{eq:executed_flow_all} cannot be computed, this method still relies on computing the estimated open-loop flow. However, by ensuring that a tube centered on $\phihatfullb$ contains $\phihatfullbE$, and that a tube centered on $\phihatfullbE$ contains $\phitruefullbE$, we may obtain tightening terms that connect the open-loop estimate to the closed-loop true state.

\subsection{Existence of Feasible Safe Controllers}

Consider a set based on the closed-loop backup flow of the true system given by \eqref{eq:executedTruth}:
\begin{align} \label{eq: CE}
    \mathcal{C}_{\rm E}(t) \triangleq \left\{ \x \in \mathcal{X} \,\middle|\, 
    \begin{array}{c}
    h(\phitruefullbE) \geq 0, \forall \nspace{1} \tau \in [0,T], \\
    h_{\rm b}(\phitruefullbET) \geq 0 \\
    \end{array}
    \right\},
\end{align}
where ${\Ce \subseteq \Cs}$. As in \Cref{sec:method1}, we have defined a set which is a based on the true state, and is therefore unknown at all times. However, we can show that, by construction, this set is controlled invariant.

\begin{theorem}[Controlled Invariance of $\Ce$] \label{thm:CE_controlledInv}
    For the output feedback system \eqref{eq:cl_outputFeedback} satisfying Assumptions~\ref{ass: eb_inCB} and~\ref{ass: robust_inv_inputdist}, the set ${\mathcal{C}_{\rm E}(t)}$ is controlled invariant\footnote{A time-varying set ${\mathcal{C}(t) \!\subset\! \mathbb{R}^n}$ is controlled invariant if a controller $\mathbf{k}\!:\!\mathcal{X} \!\times\! \R \!\rightarrow \mathcal{U}$, $\mathbf{u}\!=\!\mathbf{k}(\mathbf{x},t)$ exists which renders $\mathcal{C}(t)$ forward invariant.} and the output feedback backup controller ${\mathbf{u} = \ub(\xhat)}$ renders ${\mathcal{C}_{\rm E}(t)}$ forward invariant\footnote{\cite[Def. 4.10]{blanchini_set-theoretic_2015} A time-varying set ${\mathcal{C}(t) \subset \mathbb{R}^n}$ is forward invariant along \eqref{eq:true-CLsystem} if for all $t_0$, ${\mathbf{x}(t_0) \in \mathcal{C}(t_0)\! \implies \!\mathbf{x}(t) \in \mathcal{C}(t)}$ for all ${t \geq t_0}$.} along \eqref{eq:true-CLsystem} such that 
    \begin{align}
        \x_0 \in \mathcal{C}_{\rm E}(0) \implies \boldsymbol{\phi}_{\rm b}(t,\x_0) \in \Ce, \forall \nspace{1} t \geq 0.
    \end{align}
    \begin{proof}
        From the definition of ${\Ce}$ and with \Cref{lemma: cb_fwd_inv}
        \begin{align} \label{eq:intermed1_ch6}
            \x_0 \in \mathcal{C}_{\rm E}(0) \implies \boldsymbol{\phi}_{\rm b}(\tau,\x_0) \in \Cb \subseteq \Cs, \forall \nspace{1}\tau \geq T. 
        \end{align}
        By the recursive nature of the flow, for any ${\x \in \Rn}$ and ${\tau,t \geq 0}$, ${\boldsymbol{\phi}_{\rm b}(\tau + t,\x) = \boldsymbol{\phi}_{\rm b}(\tau,\boldsymbol{\phi}_{\rm b}(t,\x))}$.
        Using \eqref{eq:intermed1_ch6} and the flow's recursive property we have 
        \begin{align} \label{eq:intermed2_ch6}
            \x_0 \in \mathcal{C}_{\rm E}(0) \implies \boldsymbol{\phi}_{\rm b}(T,\boldsymbol{\phi}_{\rm b}(t,\x_0)) \in \Cb, \forall \nspace{1}t \geq 0. 
        \end{align}
        From \eqref{eq:intermed1_ch6} and from the definition of ${\Ce}$ in \eqref{eq: CE}, ${\x_0 \in \mathcal{C}_{\rm E}(0) \implies \boldsymbol{\phi}_{\rm b}(\tau,\x_0) \in \Cs, \forall \nspace{1} \tau \geq 0}$. Using the recursive nature of the flow once again
        \begin{align} \label{eq:intermed3_ch6}
            \x_0 \!\in\! \mathcal{C}_{\rm E}(0) \!\!\implies\!\! \boldsymbol{\phi}_{\rm b}(\tau,\boldsymbol{\phi}_{\rm b}(t,\x_0)) \!\in\! \Cs, \nspace{-3}\forall \nspace{1}\tau \!\in\! [0,T], \nspace{-2}\forall \nspace{1}t \geq 0. \nspace{36}\raisetag{.83\baselineskip}
        \end{align}
        Thus, \eqref{eq:intermed2_ch6} and \eqref{eq:intermed3_ch6} satisfy the definition of ${\Ce}$ in \eqref{eq: CE} which completes the proof.
    \end{proof}
\end{theorem}

This gives a stronger result about the unknown set than in \Cref{thm: CI_controlledinv_kinda} because we can guarantee that all true states which belong to ${\Ce}$ can remain in ${\Ce}$. However, because this set is unknown, the control design cannot use such a set, so there are no significant practical implications of this stronger theoretical result. 
We now inform the design of the tightening terms with the following Lemma.

\begin{lemma}[Conditions on $\Cihat(t)$ Imply Safety on $\Ce$] \label{lemma: pseudoSubsetCE}
    Let ${\hat{\delta}(\tau,t)}$ be a norm bound on the deviation between ${\phihatfullb}$ and ${\phitruefullbE}$ at time ${\tau \in [0,T]}$ and global time ${t}$, for all ${\mathbf{x} \in \mathcal{X}}$ and ${\xhat \in \Xhat}$:
    \begin{equation} \label{eq: bound_method2_requirement}
        \big\|{\underbrace{\phitruefullbE - \phihatfullb}_{\triangleq \nspace{2} \Delta\varphi}}\big \| \leq \hat{\delta}(\tau,t).
    \end{equation}
    If the tightening terms defining $\Cihat(t)$ in \eqref{eq: CI_hatdef} satisfy
    \begin{subequations} \label{eq:sup_epsilon_2}
        \begin{flalign}
           \!& \epsilon_\tau  \! \geq  \!\!\!\!\sup_{\norm{{\Delta \varphi}}\leq\hat{\delta}(\tau,t)} \!\!\!\! h(\phihatfullb) \!- \! h( \phihatfullb \! +  \!{\Delta \varphi}), \label{eq:supremum_epsilon_tau_2} \\
           \!& \epsilon_{\rm b} \! \geq  \!\!\!\!\sup_{\norm{{\Delta \varphi}}\leq\hat{\delta}(T,t)} \!\!\!\! h_{\rm b}(\phihatfullbT) \! - \! h_{\rm b}( \phihatfullbT \! +  \!{\Delta \varphi}), \label{eq:supremum_epsilon_b_2} \nspace{50} \raisetag{1.5\baselineskip}
        \end{flalign}
    \end{subequations}
    for all ${\tau  \in [0,T]}$ and ${t \geq 0}$, then ${\xhat \in \Cihat(t) \implies \x \in \mathcal{C}_{\rm E}(t)}$.
    \begin{proof}
        The proof is analogous to \Cref{lemma: pseudoSubsetCi}.
    \end{proof}
\end{lemma}
Note that the observations and discussions made in \Cref{rem:tighteningTerms} on how to compute the supremum terms in \eqref{eq:sup_epsilon_2} for different forms of barriers apply to \Cref{lemma: pseudoSubsetCE} as well.
We now present a result which is analogous to \Cref{thm:feasible_Cihat} for the closed-loop method.

\begin{theorem}[Existence of a Feasible and Safe Controller in $\Cihat(t)$] \label{thm:feasible_Cihat_method2}
If ${\epsilon_{\tb}}$ and ${\epsilon_{\rm b}}$ satisfy \eqref{eq:sup_epsilon_2}
for all ${{\tau \in [0,T]}}$ and ${{t \geq 0}}$, with ${\hat{\delta}(\tb,t)}$ defined in~\eqref{eq: bound_method2_requirement}, then for any ${\xhat_0 \in \Cihat(0)}$, there exists a controller ${\mathbf{k} : \Xhat \rightarrow \mathcal{U}}$, ${\mathbf{u} = \mathbf{k}(\xhat)}$ for system \eqref{eq:true-CLsystem} such that ${\x(t) \in \mathcal{C}_{\rm E}(t)\subseteq \Cs}$, for all ${t \geq 0}$.
\begin{proof}
    By \Cref{lemma: pseudoSubsetCE}, ${\xhat_0 \in \Cihat(0) \implies \x_0 \in \mathcal{C}_{\rm E }(0)}$, and by \Cref{thm:CE_controlledInv} the backup controller ${\ub}$ ensures that 
    ${\x(t) \in \mathcal{C}_{\rm E}(t)\subseteq \Cs, \forall \nspace{1} t \geq 0}$.
\end{proof}
\end{theorem}

\subsection{Tightening Terms and Closed-Loop Flow Bounds}

As discussed in \Cref{rem: method1-bound-bad}, the bound in \Cref{lemma: flowbound_method1_ch6} may be impractical for \textit{general} nonlinear systems, especially because it relies on the Lipschitz constants ${\mathcal{L}_{\mathbf{f}}}$ and ${\mathcal{L}_{\mathbf{g}}}$ of the dynamical system. These constants are inherent to the system, giving us no control over them (besides restricting the domain of operation, i.e., ${\mathcal{X},\Xhat}$). To address this drawback, we can instead derive a bound as required by \Cref{lemma: pseudoSubsetCE} which relies on characterizations of the \textit{closed-loop state feedback} backup dynamics: ${\mathbf{f}_{\rm cl}(\x) = \mathbf{f}(\x) + \mathbf{g}(\x)\ub(\x)}$.

\begin{lemma}[Closed-loop Flow Bounds] \label{lemma: flowbound_method2_ch6}
Let the closed-loop backup dynamics ${\mathbf{f}_{\rm cl}}$ in \eqref{eq: computed_flow_est} and \eqref{eq:executedEst} be one-sided Lipschitz\footnote{See \cite[Ch. 3]{contraction_bullo} for details.} with constant ${\kappa_{\rm cl}\!\in\! \R}$ and the measurement function ${\mathbf{z}}$ in \eqref{eq:measFun} be Lipschitz continuous with constant ${\mathcal{L}_{\mathbf{z}}}$. If the estimator takes the linear correction form as in \Cref{ex:linear_correction}, and ${\norm{\mathbf{L}(t)} \!\leq\! \Bar{L}}$ for all ${t \!\geq\! 0}$, then for all ${\tau \!\in\! [0,T]}$ and ${t \!\geq\! 0}$ the following bound satisfies \eqref{eq: bound_method2_requirement}
\begin{multline} \label{eq:bound_method2}
\hat{\delta}(\tau,t) \triangleq \delta_x(t + \tau, \Bar{e}_0) + \frac{\Bar{L} \Bar{v}}{\kappa_{\rm cl}} \big({\rm e}^{\kappa_{\rm cl}\tau} - 1\big) \\ + \Bar{L} \mathcal{L}_{\mathbf{z}} \int^\tau_0 {\rm e}^{\kappa_{\rm cl}(\tau - s)} \delta_x(t+s,\Bar{e}_0) \nspace{1} ds,
\end{multline}
for a time-varying function 
${\delta_x}$ 
from \Cref{def: error_estimator}.
\begin{proof}
    Recalling ${\Delta\varphi(\tau,\x,\xhat)  \triangleq \phitruefullbE - \phihatfullb}$, dropping the dependence on its arguments,  adding zero and applying the triangle inequality
    we have
    \begin{align*}
        \|{\Delta\varphi}\| &= \|{\big(\phihatfullbE - \phihatfullb \big) + \big( \phitruefullbE - \phihatfullbE  \big)}\|, \\
        &\leq \|{\phihatfullbE - \phihatfullb }\| +  \|{\phitruefullbE - \phihatfullbE}\|.
    \end{align*}
    Since ${\|{\phitruefullbE - \phihatfullbE }\|}$ represents the difference between the solutions to ${\eqref{eq:true-CLsystem}}$ and \eqref{eq:est-CLsystem} at ${t + \tau}$, we have
    \begin{align} \label{eq: intermed_method2_boundproof}
        \norm{\Delta\varphi}
        \leq \|{\phihatfullbE - \phihatfullb}\| +  \delta_x(t+\tau,\Bar{e}_0).
    \end{align}
    The remainder of the proof is focused on bounding the first term on the right-hand side of \eqref{eq: intermed_method2_boundproof}. To do this, we leverage the properties of ${\mathbf{f}_{\rm cl}}$. 
    Examining \eqref{eq:executedEst}, one can immediately notice that this is dynamical system \eqref{eq: computed_flow_est} with an additive disturbance ${\mathbf{d}_\Delta \!\triangleq\! \mathbf{L}(t \!+\! \tau)\big(\mathbf{y}(t \!+\! \tau) \!-\! \mathbf{z}(\phihatfullbE)\big)}$. The disturbance is expanded as
    \begin{align*}
        \mathbf{d}_\Delta = \mathbf{L}(t + \tau)\big(\mathbf{z}(\phitruefullbE) - \mathbf{z}(\phihatfullbE) + \mathbf{v}(t+\tau)\big).
    \end{align*}
    By Lipschitz continuity of the measurement function ${\mathbf{z}}$, and the fact that ${\norm{\mathbf{L}(t)} \leq \Bar{L}}$ and ${\norm{\mathbf{v}(t)} \leq \Bar{v}}$ for all ${t \geq 0}$, this disturbance can be bounded such that
    \begin{align}
        \norm{\mathbf{d}_\Delta \nspace{-1}} 
        \leq \Bar{L}\big(\mathcal{L}_{\mathbf{z}} \delta_x(t+\tau,\Bar{e}_0) + \Bar{v}\big). \label{eq:disturbance_ch6_method2}
    \end{align}
    Applying \cite[Corollary 3.17]{contraction_bullo} yields
    \begin{align*}
        \norm{\phihatfullbE \!-\! \phihatfullb } \!\leq\! \Bar{L}\! \!\int^\tau_0 \!\!\!{\rm e}^{\kappa_{\rm cl}(\tau - s)} \big(\mathcal{L}_{\mathbf{z}} \delta_x(t\!+\!s,\Bar{e}_0) \!+\! \Bar{v}\big) ds,
    \end{align*}
    and integrating completes the proof.
\end{proof}
\end{lemma}
\begin{remark}[Innovation as a Decaying Disturbance]
    Because in the long term the estimation error is expected to decay, the ``disturbance'' bound considered in \eqref{eq:disturbance_ch6_method2} also decays,
    which may result in a practical (possibly tight) overall bound. We also note that in the case of a negative constant $\kappa_{\rm cl}$, the backup dynamics are said to be contracting \cite{contraction_lcss_bullo24}, yielding even tighter flow bounds. 
\end{remark}

With \Cref{thm:feasible_Cihat_method2}, it is trivial to show that the conditions in \Cref{thm: fwd_inv_Cihat} render ${\Cihat(t)}$ forward invariant for ${\epsilon_\tau}$ and ${\epsilon_{\rm b}}$ given in \Cref{lemma: pseudoSubsetCE}. By the same arguments then, solving the \eqref{eq:e-bcbf-qp} and changing the ${\epsilon_{\rm b}}$, ${\epsilon_\tau}$, ${\dot{\epsilon}_{\rm b}}$, and ${\dot{\epsilon}_\tau}$ terms in \eqref{eq: mainthmConstraints_ch6} yields a control signal which is output feedback safe (by \Cref{def:outputSafety}) for all ${\x_0,\xhat_0}$ satisfying ${\Bar{e}_{0} \!\geq\! \norm{\x_0 \!-\! \xhat_0}}$ and ${\xhat_0 \!\in\! \Cihat(0)}$. In this case though, ${\xhat \!\in\! \Cihat(t) \!\implies\! \x \!\in\! \Ce\! \subseteq\! \Cs}$. 
\begin{remark} 
[Applicability of Methods]
As emphasized throughout the paper, 
the closed-loop and open-loop methods are similar, differing primarily in the form of the flow bounds used to relate the true and estimated states (i.e., \Cref{lemma: flowbound_method1_ch6}, \Cref{lemma: flowbound_method2_ch6}). The open-loop method is comparatively simpler, and results in reasonable flow bounds for linear systems or systems with an input map $\mathbf{g}$ that is constant or has a small Lipschitz constant. In contrast, the closed-loop method may be particularly useful when $\mathbf{f}_{\rm cl}$ can be made contracting, or when $\mathbf{g}$ is highly nonlinear.
\end{remark}
\section{Numerical Examples} \label{sec:examples}

We now demonstrate the efficacy of the presented techniques using two simulation examples\footnote{Code and videos for all simulation examples may be found at \href{https://github.com/davidvwijk/OutputFeedbackCBF}{\texttt{https://github.com/davidvwijk/OutputFeedbackCBF}}.}.

\begin{casestudy}[Double Integrator] \label{ex:db_int}
Consider the double integrator system with
\begin{subequations} \label{eq:CL_eqtns_dbint_ch6}
\begin{align}
    &\dot{\mathbf{x}} =
    \big[x_2\nspace{5} u \big]^\top, \label{eq: truth_dbInt} \\
    &\dot{\hat{\x}} = 
    \big[\hat{x}_2\nspace{5} u \big]^\top \!+ \mathbf{L}\big({y} - \C\hat{\x}\big),  \label{eq: est_dbInt} \\
    &\x(0) = \x_0, \nspace{4} \xhat(0) = \xhat_0,
\end{align}
\end{subequations}
with true state $\mathbf{x} \!=\! [x_1\nspace{5} x_2]^\top \!\in\! \X\! \subset \!\mathbb{R}^2$, estimate $\xhat \!= \![\hat{x}_1\nspace{5} \hat{x}_2]^\top \!\in\! \Xhat \!\subset\! \mathbb{R}^2$ and control input $u \!\in\! \mathcal{U} = [-u_{\rm max}\nspace{5} u_{\rm max}]$.
Assume that noisy measurements of ${x_1}$ (e.g., range measurements) can be continuously obtained such that 
\begin{align}
    {y} &= \mathbf{C}\mathbf{x} + {v}(t), \quad \mathbf{C} = \big[1 \nspace{10} 0\big],
\end{align}
with ${|{{v}(t)}| \!\leq\! \Bar{v} \!\in\! \R_{> 0}}$.
Consider the safe set $\Cs \!\triangleq\! \{\mathbf{x}\! \in \!\mathbb{R}^2 \!:\! h(\x) \!\geq\! 0\}$ with ${h(\x) \!\triangleq\! x^2_{\rm max} \!-\! x_1^2}$ for a constant ${x_{\rm max} \!\in\! \R_{> 0}}$.

We utilize the backup set described by \eqref{eq: backupSetLinear_general} in \Cref{prop:linearsystemCBandub}.
Using the condition given in \Cref{prop:linearsystemCBandub}, the gain of the backup controller ${\ub(\xhat) \!=\! - \K \xhat}$ can be selected such that this output feedback controller renders ${\Cb}$ forward invariant for \eqref{eq: truth_dbInt}. Furthermore, if the gain is selected such that ${\sqrt{\gamma} \norm{\K \mathbf{P}^{-1/2}} \!+\! \norm{\K}\Bar{e}_{\rm b} \!\leq\! u_{\rm max}}$, then the backup controller ${\ub(\xhat) = u_{\rm max}\nspace{2}{\rm sat}{(-\K\xhat)}}$ will not saturate\footnote{This can be verified using \cite[(B.4)]{khalil_nonlinear_2015}.} for any $\x$ inside $\Cb$. Here, $\rm sat(\cdot)$ is a smooth saturation function, which we approximate with a hyperbolic tangent in the simulations. For the quadratic barrier defined in \eqref{eq: backupSetLinear_general}, the supremum term in \eqref{eq:supremum_epsilon_b_1} is exactly
\begin{align} \label{eq:intermediate_dbInt}
\sup_{\norm{{\Delta \phi}}\leq\hat{\delta}(T,t)} {\Delta \phi}^\top \mathbf{P}  {\Delta \phi} + 2 \phihatfullbT^\top \mathbf{P} {\Delta \phi},    
\end{align}
and can be computed numerically, or overapproximated by
\begin{align*}
   \epsilon_{\rm b} = \hat{\delta}(T,t)^2 \lambda_{\rm max}(\mathbf{P}) + 2\hat{\delta}(T,t) \big\|{\mathbf{P}\phihatfullbT}\big\| \geq \eqref{eq:intermediate_dbInt}.
\end{align*}
Note that because in this case $\epsilon_{\rm b}$ has an $\xhat$ dependence, the $\dot{\epsilon}_{\rm b}$ term in \eqref{eq: mainthmConstraints_reach_ch6} would include a time derivative of $\xhat$. For simplicity, these terms are neglected in the simulations.

We simulate\footnote{The simulation uses the constants: $u_{\rm max} \!=\! 2$, $k_{\rm p}(\xhat,t) \!=\! u_{\rm max}\nspace{2}{\rm sin}{(t)}$, $\K \!=\! [1.535\nspace{5} 1.382]$, ${x_{\rm max} \!=\! 2}$, ${\mathbf{L} \!=\! [2\nspace{5} 2]^\top}$, ${\Bar{v} \!=\!  0.02}$, ${\ebar \!=\! 0.2}$, ${\Bar{e}_{\rm b} \!=\! 0.15}$, ${\mathbf{Q} \!=\! \mathbf{I}_2}$, ${\gamma \!=\! 0.76}$, ${\Delta \!=\! 0.02 \nspace{4} {\rm sec}}$ and ${T \!=\! 2 \nspace{4} {\rm sec}}$, ${\alpha(r) = 10r + r^3}$, ${\alpha_{\rm b}(r) = 10r}$.} system \eqref{eq:CL_eqtns_dbint_ch6} using the \eqref{eq:e-bcbf-qp}, and utilize the flow bound ${\hat{\delta}(\tau,t)}$ given in \eqref{eq: linearBound_method1} to design the tightening terms ${\epsilon_\tau}$ and ${\epsilon_{\rm b}}$. The estimation error bound ${\delta_x(t,\ebar)}$ is given by \eqref{eq:delx_linear}. \Cref{fig:db_int_case2_ch6_PhasePlot} plots the trajectory of the estimate (solid pink) and the true state (solid black) in the phase space, along with the estimation error bounds (gray) centered around the estimate trajectory. As per the definition of ${\Cb}$, the backup set is an ellipsoidal region centered on the origin (cyan). Despite the true state remaining unknown, the proposed approach achieves output feedback safety using the estimate, ${\xhat(t)}$, and examining the forward evolution of the state estimate using the open-loop method described in \Cref{sec:method1}. The dashed pink circles are the open-loop flow bounds in \eqref{eq: linearBound_method1}, and note that only a few open-loop flows are plotted to reduce clutter, but in the simulation these are propagated at every time step. Initially, due to large state estimation error bounds, the open-loop flow envelope (dashed pink circles) is large, but as the estimation error bounds shrink, so do the flow envelopes, reducing conservatism as global time increases. 

\begin{figure}[t]
    \centering
    {\includegraphics[width=1\linewidth]{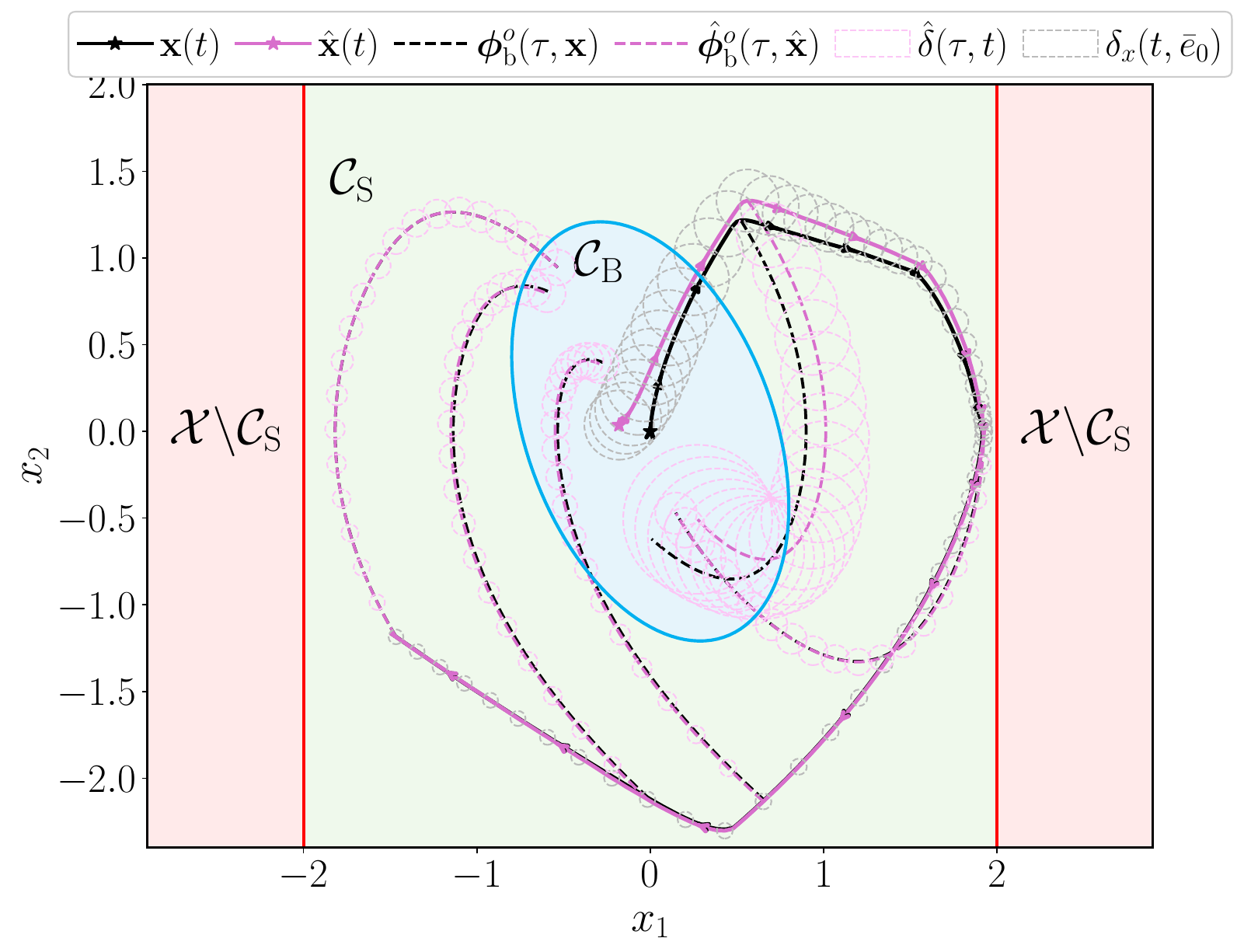}}
    \vspace{-.7cm}
    \caption{Double integrator system \eqref{eq:CL_eqtns_dbint_ch6} under safe output feedback, using the proposed output feedback backup control barrier function approach.
    The true system (solid black) remains in ${\Cs}$ using only knowledge of the estimated state (solid pink). The open-loop estimated backup trajectories and corresponding uncertainty envelope (dashed pink) are used to construct ${\Cihat(t)}$ and by ensuring ${\xhat(t) \!\in\! \Cihat(t)}$, safety of ${\x(t)}$ is guaranteed by \Cref{thm: fwd_inv_Cihat}. See animation at \href{https://youtu.be/q5r6rP1gVlQ}{\footnotesize\texttt{https://youtu.be/q5r6rP1gVlQ}}.  
    }
    \label{fig:db_int_case2_ch6_PhasePlot}
    \vskip -2mm
\end{figure}
\begin{figure}
    \centering
    \includegraphics[width=1\linewidth]{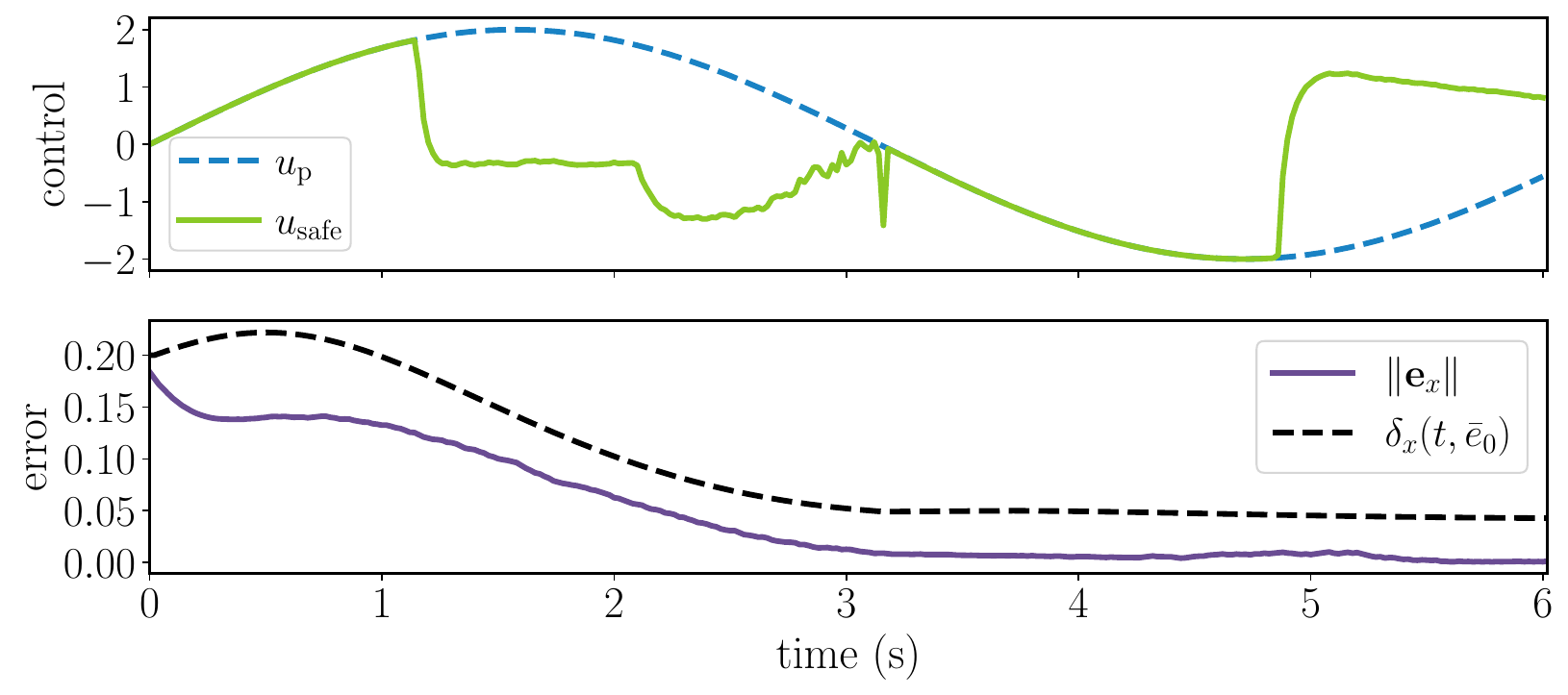}
    \vspace{-.7cm}
    \caption{The safe control input (green) satisfies input constraints, whilst minimizing the deviation from the primary controller (dashed blue) ({\textbf{top}}). State estimation error bounds and control input for double integrator system \eqref{eq:CL_eqtns_dbint_ch6} under output feedback. The estimation error always satisfies ${\norm{\mathbf{e}_x} \leq \delta_x(t,\Bar{e}_0)}$ as guaranteed by \Cref{prop:linearBoundEstimation} ({\textbf{bottom}}).}
    \label{fig:db_int_case2_ch6_BiPlot}
    \vskip -4mm
\end{figure}

\Cref{fig:db_int_case2_ch6_BiPlot} plots the norm of the state estimation error (bottom) and the control input (top). While the estimation error (solid purple) is initially nearly as large as ${\ebar}$, the error shrinks rapidly with the accumulation of measurements, and the upper bound on ${\norm{\mathbf{e}_x}}$ is maintained throughout the simulation as guaranteed by \Cref{prop:linearBoundEstimation}. As evidenced by the safe control input (solid green), the ${\eqref{eq:e-bcbf-qp}}$ intervenes when necessary 
all whilst adhering to the input bounds imposed by $u_{\rm max}$.
\end{casestudy}

\begin{casestudy}[Spacecraft Rotation]
\begin{figure}[t]
    \centering
    \includegraphics[width=1\linewidth]{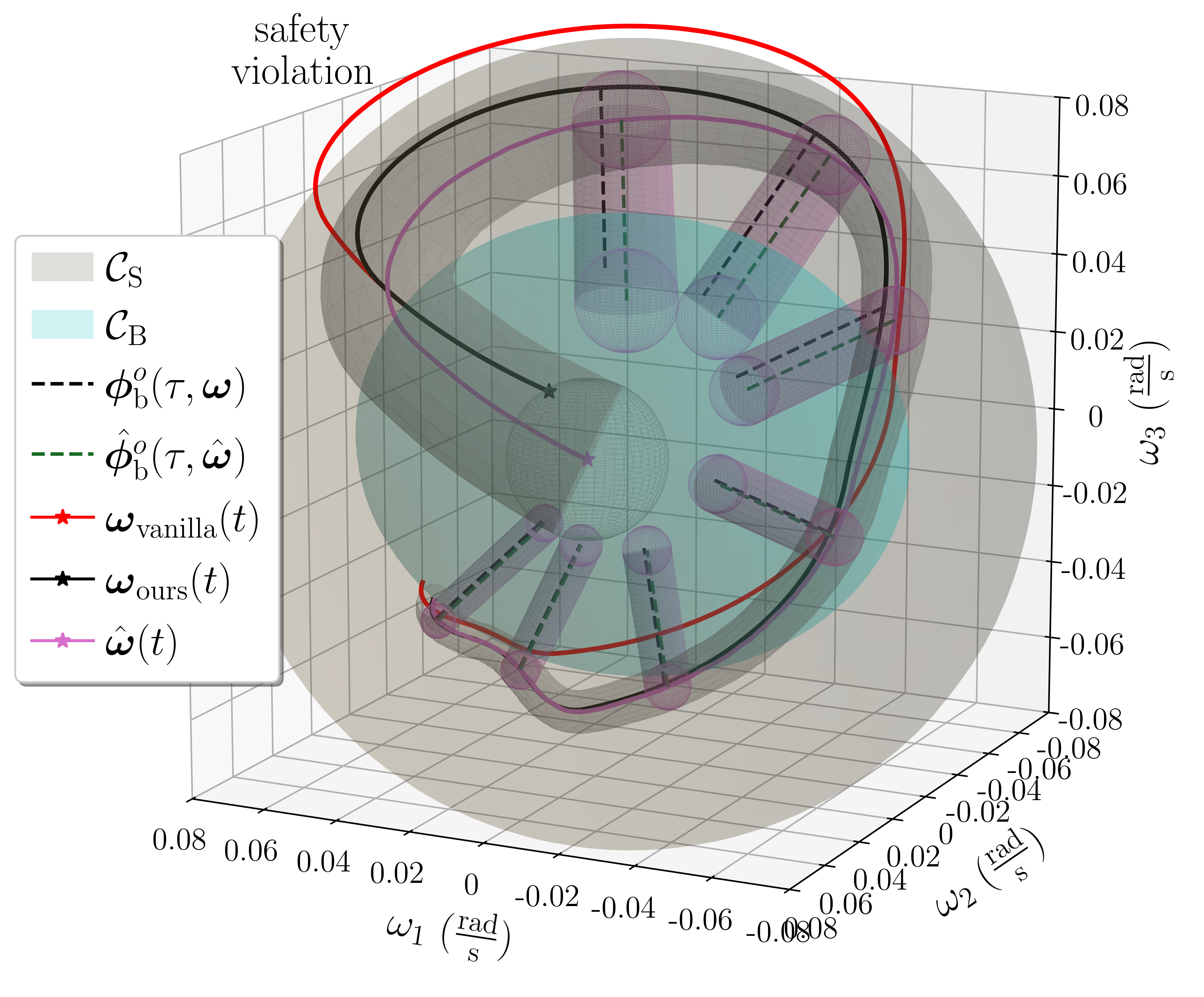}
    \vspace{-.7cm}
    \caption{State space visualization of angular velocity components showing the trajectory of the angular velocity vector over time for the spacecraft model \eqref{eq:CL_eqtns_spacecraft_ch6}. The safety objective is to keep the true angular velocity, ${\w(t)}$, within the sphere for all ${t\geq0}$, using only information about the estimated trajectory, ${\what(t)}$ in pink. The standard approach reviewed in \Cref{sec:bCBF} fails to achieve safety in the presence of estimation error (red) while our approach (black) achieves safety and recursive feasibility. The tube centered around the estimated angular velocity in pink represents the estimation error bound ${\delta_x(t,\Bar{e}_0)}$. Recursive feasibility in the presence of input bounds is achieved by propagating the open-loop estimated flow (dashed green) and ensuring the open-loop flow bounds, ${\hat{\delta}(\tau,t)}$ shown by the pink tubes, are contained in ${\Cs}$ and are contained in $\Cb$ at ${\tau = T}$. Note that for visual clarity, only a few such open-loop flows and bounds are plotted, though in simulation these are computed at each timestep.}
    \vskip -4mm
    \label{fig:3d_plot_spacecraft}
\end{figure}
\begin{figure}[t]
    \centering
    \includegraphics[width=1\linewidth]{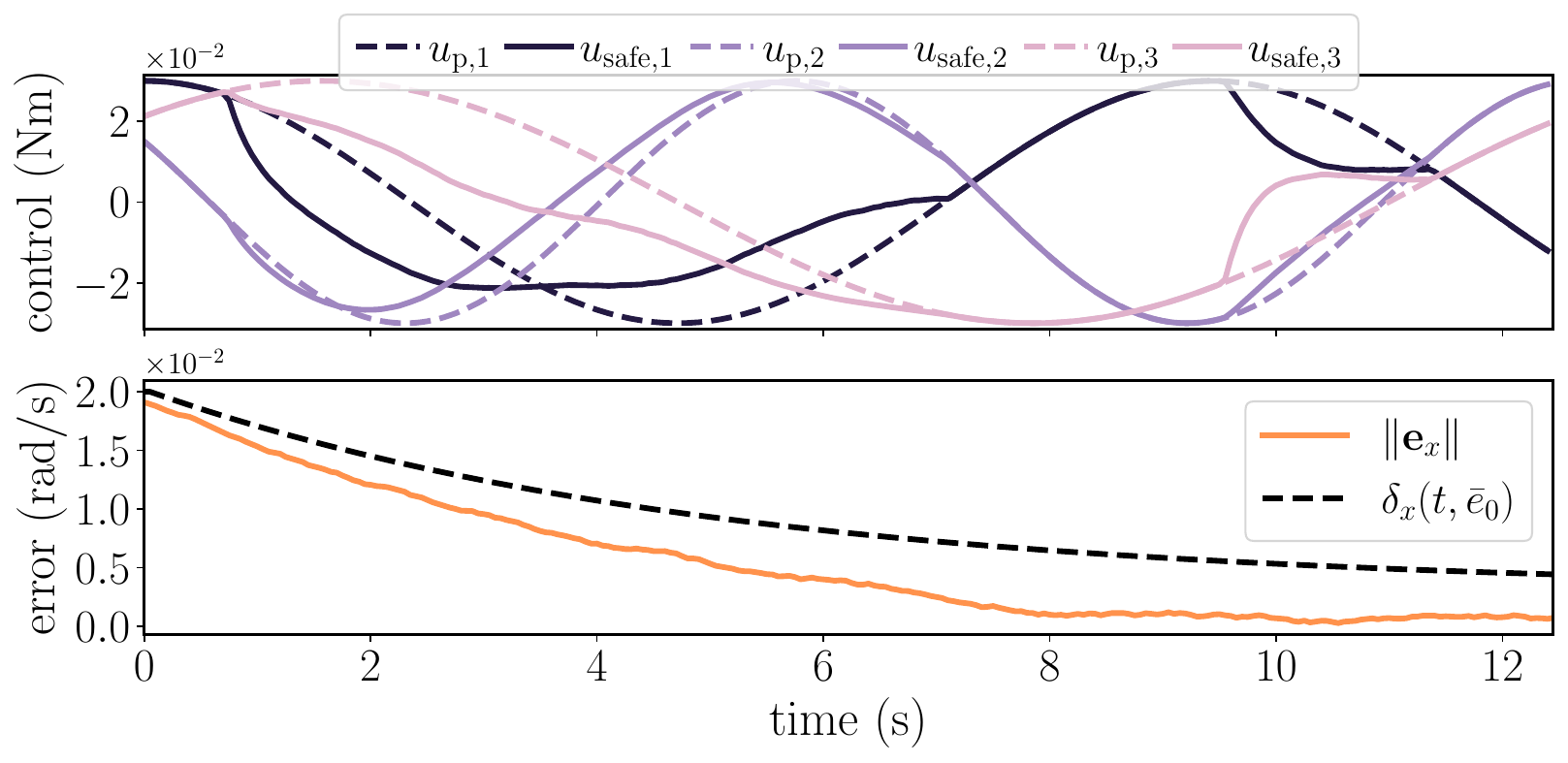}
    \vspace{-.7cm}
    \caption{
    Desired control (dashed) and safe control signal from \eqref{eq:e-bcbf-qp} (solid) over time for each axis (\textbf{top}). Thanks to \Cref{thm:feasible_Cihat}, there always exists a safe output feedback control signal which satisfies the input bounds imposed by ${\mathcal{U}}$. The estimation error (solid orange) satisfies ${\norm{\mathbf{e}_x} \leq \delta_x(t,\Bar{e}_0)}$ (\textbf{bottom}).}
    \vskip -4mm
    \label{fig:subplots_spacecraft}
\end{figure}

Consider next an example of a rigid body spacecraft \cite[Example 2.4.5]{mote2021optimization} with a known inertia tensor in the body frame given by $\J$, with Euler's rotation equations of motion given by 
\begin{subequations} \label{eq:CL_eqtns_spacecraft_ch6}
\begin{align}
    &\dot{\boldsymbol{{\omega}}} = \mathbf{J}^{-1}(-\w\times\mathbf{J}\boldsymbol{\omega} + \mathbf{u}), \label{eq: truth_sc} \\
    &\dot{\widehat{\boldsymbol{{\omega}}}} = \mathbf{J}^{-1}(-\what\times\mathbf{J}\what + \mathbf{u}) + \mathbf{L}(t)\big(\mathbf{y} - \mathbf{z}(\what) \big),  \label{eq: est_sc} \\
    &\w(0) = \w_0, \nspace{6} \what(0) = \what_0,
\end{align}
\end{subequations}
where ${\w \in \mathbb{R}^3}$ is the true angular velocity in the body frame, ${\what \in \mathbb{R}^3}$ is the angular velocity estimate in the body frame, the control input is ${\mathbf{u}\!\in\! \mathcal{U} \!=\! [-u_{\rm max}\nspace{5} u_{\rm max}]^3}$ and the initial conditions satisfy ${\norm{\w_0 \!-\! \what_0} \!\leq\! \ebar}$. The spacecraft estimates its angular velocity using \eqref{eq: est_sc} with noisy gyroscope measurements such that ${    \mathbf{y} = \mathbf{C}\w + \mathbf{v}(t)}$,
with ${\norm{\mathbf{v}(t)} \leq \Bar{v}}$ and $\mathbf{C} = \mathbf{I}_{3}$. The time-varying gain, ${\mathbf{L}(t)}$, is taken to be the gain of the continuous-time EKF \cite[Section 11.2]{khalil_nonlinear_2015} such that ${\mathbf{L}(t) = \boldsymbol{\Sigma}(t)\mathbf{C}^\top(t) \mathbf{R}^{-1}}$,
where ${\boldsymbol{\Sigma}(t)}$ is the solution to the differential Riccati equation
\begin{align*}
    \dot{\boldsymbol{\Sigma}} = \mathbf{F} \boldsymbol{\Sigma} + \boldsymbol{\Sigma} \mathbf{F}^\top + \mathbf{W} - \boldsymbol{\Sigma}\mathbf{C}^{\top}\mathbf{R}^{-1}\mathbf{C}\boldsymbol{\Sigma}, \quad \boldsymbol{\Sigma}(0) = \boldsymbol{\Sigma}_0.
\end{align*}
$\mathbf{F}$ is the Jacobian of \eqref{eq: truth_sc} evaluated at $\what$, and the matrices ${\boldsymbol{\Sigma}_0,\mathbf{W},\mathbf{R}}$ are symmetric and positive definite. We assume that the estimation error is upper bounded by an exponentially decaying function where $\delta_x(t, \Bar{e}_0) \!=\! \Bar{e}_0 \!-\! \beta (1 \!-\! {\rm e}^{- \kappa t}) \! \geq\! \norm{\w(t) \!-\! \what(t)}$ for all ${t\!\geq\!0}$ and some ${\beta,\kappa \!>\! 0}$.

The safety objective is to ensure that the norm of the angular velocity vector of the spacecraft does not exceed a maximum value, to prevent damage to the spacecraft. The safe set is therefore ${\mathcal{C}_{\rm S} = \{\boldsymbol{\omega} \in \mathbb{R}^{3} : h(\boldsymbol{\omega}) \ge 0\}}$ where ${h(\boldsymbol{\omega}) \triangleq \omega_{\rm max}^2 - \left \Vert \boldsymbol{\omega} \right \Vert^2}$ and ${\omega_{\rm max} \in \mathbb{R}_{>0}}$ represents the maximum allowable angular velocity. The primary controller, intentionally designed to violate safety, is given by ${\boldsymbol{k}_{\rm p}(\what,t) = u_{\rm max}\text{cos}([\frac{t}{1.5}\nspace{10} \frac{t}{1.1} + \frac{\pi}{3}\nspace{10} \frac{t}{2} - \frac{\pi}{4}]^\top)}$.

The backup set ${\Cb \!\triangleq\! \{ \w \!\in\! \mathbb{R}^{3} \!:\!  h_{\rm b}(\w) \!=\! \gamma \!-\! \frac{1}{2}\w^\top\J\w \!\ge\! 0\}}$ is a level set of the spacecraft's rotational energy, defined by the positive scalar ${\gamma > 0}$, where $\gamma$ is chosen such that ${\Cb \subseteq \Cs}$. The state feedback backup control law $\ub(\w) = -K_{\rm b}\J\w + \w \times \J \w$ can be shown to satisfy \Cref{ass: robust_inv_inputdist} for gain $K_{\rm b}$ satisfying 
\begin{align*}
    K_{\rm b} \geq \frac{2\lambda_{\rm max}(\J)\norm{\J}\norm{\J^{-1}}\Bar{e}_{\rm b}\omega_{\rm max}}{\sqrt{2\gamma\lambda_{\rm min}(\J)} - \lambda_{\rm max}(\J)\norm{\J}\norm{\J^{-1}}\Bar{e}_{\rm b}}.
\end{align*}
Therefore, by \Cref{lemma: cb_fwd_inv}, the output feedback backup control law ${\ub(\what) = -K_{\rm b}\J\what + \what \times \J \what}$ renders the backup set $\Cb$ forward invariant along \eqref{eq: truth_sc}. Further, ${\ub}$ will not violate input constraints within $\Cs$ as long as ${u_{\rm max}}/({\norm{\J}\omega_{\rm max}}) - \omega_{\rm max} \geq K_{\rm b}$.

We simulate system \eqref{eq:CL_eqtns_spacecraft_ch6} to demonstrate the efficacy of the \eqref{eq:e-bcbf-qp} in comparison with the standard backup CBF approach reviewed in \Cref{sec:bCBF} which directly uses the estimated state in place of the true state\footnote{The simulations use the following constants: ${u_{\rm max} = 0.03\nspace{4} {\rm Nm}}$,  ${\omega_{\rm max} = 0.1\nspace{4} {\rm rad/s}}$, ${K_{\rm b} = 0.2746 \nspace{4} {\rm s^{-1}}}$, ${\Bar{v} = 0.01  \nspace{4} {\rm rad/s}}$, ${\ebar = 0.02 \nspace{4} {\rm rad/s}}$, ${\Bar{e}_{\rm b} = 0.01 \nspace{4} {\rm rad/s}}$, ${\beta = 0.017 \nspace{4} {\rm rad/s}}$, ${\kappa = 0.2 \nspace{4} {\rm s^{-1}}}$, ${\gamma = 0.0013 \nspace{4} {\rm J}}$, ${\Delta = 0.05 \nspace{4} {\rm sec}}$, ${T = 3 \nspace{4} {\rm sec}}$, ${\J = \texttt{diag}\{[0.5186 \nspace{10}  0.8006 \nspace{10} 0.8006]\} \nspace{4} {\rm kg m^2}}$.}. In this scenario, the open-loop flow bound in \Cref{lemma: flowbound_method1_ch6} is used, and the tightening constants are computed using the local Lipschitz constants of $h$ and $h_{\rm b}$ (see Remark~\ref{rem:tighteningTerms}).
\Cref{fig:3d_plot_spacecraft} verifies that the true angular velocity using our approach in solid black obeys the safety constraint (i.e., is safe w.r.t. \Cref{def:outputSafety}) even though the output feedback controller only uses estimated states. Further, as evidenced by \Cref{fig:subplots_spacecraft}, the safe control signal always adheres to the input bounds imposed by ${\mathcal{U}}$. In contrast, the true angular velocity using the vanilla backup CBF approach in solid red violates the constraint by exiting the sphere, as shown in \Cref{fig:3d_plot_spacecraft}, because this method does not account for the state estimation error.
\end{casestudy}
\section{Conclusion} \label{sec:conclusion}

This manuscript developed novel techniques for maintaining output feedback safety of input-constrained systems where only an estimate of the true state is known. By designing a backup controller which renders a backup set forward invariant under output feedback, we extended the approach of backup CBFs for this context, and presented two methods for examining the flow of the open-loop estimated backup system. Since this flow can be computed, we derived forward invariance conditions for a set based on this flow, and showed that these conditions imply output feedback safety for the unknown, true state.
\section{Acknowledgments}

This work was supported by DARPA under the LINC program and by the Technology Innovation Institute (TII).

\bibliographystyle{ieeetr} 
\bibliography{refs}
\end{document}